\let\oldvec\vec% Store \vec in \oldvec
\documentclass[runningheads]{llncs}
\let\vec\oldvec% Restore \vec from \oldvec
\usepackage{xcolor}

\usepackage{amsmath,amsfonts,amssymb}
\usepackage{mathtools}
\usepackage{bm}
\usepackage{mathdots}
\usepackage{stmaryrd}
\usepackage{tikz}
\usepackage{float}
\usepackage{cite}
\usepackage{wrapfig}
\usepackage{complexity}
\usepackage{esvect}
\usepackage{xcolor}
\usepackage{etoolbox}
\usepackage{blkarray}
\usepackage{ifthen}
\usepackage{multirow}
\usepackage{thmtools}
\usepackage[firstpage]{draftwatermark}  % free badge placement
\usepackage{array}
\usepackage[ruled,vlined,linesnumbered,nofillcomment]{algorithm2e}
\usepackage{todonotes}
\SetKwComment{Comment}{$\triangleright$\ }{}

\SetCommentSty{mycommfont}
\makeatletter
\patchcmd{\@algocf@start}% <cmd>
  {-1.5em}% <search>
  {0pt}% <replace>
  {}{}% <success><failure>
\makeatother
\newcommand{\union}{\cup}
\newcommand{\Union}{\bigcup}
\newcommand{\Land}{\bigwedge}
\newcommand{\Lor}{\bigvee}
\renewcommand{\iff}{\Leftrightarrow} % shorter <=> arrow
\newcommand{\svars}[1]{\textsf{support}(#1)} % e.g. \vars( f(x,y,z) ) = {x, y, z}

\newcommand\concept[1]{\textit{#1}\xspace}

\providecommand{\tuple}[1]{\ensuremath{\left\langle #1 \right\rangle}}
\providecommand{\set}[1]{\ensuremath{\left\lbrace #1 \right\rbrace}}

\providecommand{\vect}[1]{\ensuremath{( \begin{matrix} #1 \end{matrix} )}}

\newcommand{\defn}{\,\triangleq\,}

\newcommand\con[2]{\ensuremath{#1_{|#2}}}
\renewcommand\exp{\ensuremath{\mathsf{exp}}}

\newcommand\reach{\textsc{Reach}\xspace}
\newcommand{\reachbdd}{\textsc{ReachBdd}\xspace} % the "def ReachBdd" in alg 1 needs to be changed manually
\newcommand{\reachbddpar}{\textsc{ReachBddPar}\xspace}
\newcommand{\reachmdd}{\textsc{ReachMdd}\xspace} % the "def ReachMdd" in alg 1 needs to be changed manually

\newenvironment{proofsketch}{\proof}{\endproof}

\newcommand{\fname}{\textsc} %function names in text
\SetFuncSty{fname} % function names in algorithms
\SetFuncArgSty{text}
\SetArgSty{text} % arguments of e.g. if, while statements
\DontPrintSemicolon
\SetKwIF{If}{ElseIf}{Else}{if}{then}{elif}{else then}{}
\SetKwFor{While}{while}{do}{}
\SetKwProg{Fn}{def}{ :}{end}
\let\oldnl\nl
\newcommand{\nonl}{\renewcommand{\nl}{\let\nl\oldnl}}

\tikzset{main/.style={draw,circle}} % internal bdd nodes
\tikzset{leaf/.style={draw,minimum width=1.2em,minimum height=1.2em}} % bdd leaf
\tikzset{lddnode/.style={draw,minimum width=1.2em,minimum height=1.2em}}
\tikzset{node distance=10mm}

\newcommand{\bddvar}[1]{\textsf{var}(#1)}
\newcommand{\bddlow}[1]{\ensuremath{#1[0]}}
\newcommand{\bddhigh}[1]{\ensuremath{#1[1]}}
\newcommand{\e}[2]{\ensuremath{#1[#2]}}

\newcommand{\makenode}{\fname{MakeNode}}

\newcommand{\sttuple}[2]{\ensuremath{(#1,#2)}} % notation for (source, target) state tuple
\newcommand{\sspace}{\mathbb{S}}

\let\oldparagraph\paragraph
\renewcommand{\paragraph}[1]{\oldparagraph{#1.}}

\newcommand{\aref}[1]{\hyperref[#1]{App.~\ref*{#1}}} % autoref writes Sec. for appendices
 
\emergencystretch=\maxdimen
\usepackage[caption=false]{subfig} 
\usepackage[colorlinks,citecolor=blue!60!black!70,linkcolor=blue!60!black!70]{hyperref}

\newcommand{\anonymize}[2]{#1} % use as \anonymize{non-anon text}{anon text}
\newcommand{\arxiv}[2]{#1} % use as \arxiv{arxiv text}{non-arxiv text}

\begin{document}

% \texorpdfstring command is to avoid warning form hyperref package about \thanks
\title{\texorpdfstring{A Decision Diagram Operation for Reachability}{A Decision Diagram Operation for Reachability}}
%\title{\texorpdfstring{A Decision Diagram Operation for Reachability\thanks{Supported by organization x.}}{A Decision Diagram Operation for Reachability}}

\anonymize{
\author{Sebastiaan Brand\orcidID{0000-0002-7666-2794} 
\and 
Thomas Bäck\orcidID{0000-0001-6768-1478}
\and 
Alfons Laarman\orcidID{0000-0002-2433-4174}
}
\authorrunning{S. Brand, T. Bäck, and A. Laarman}
\institute{Leiden Institute of Advanced Computer Science, \\ Leiden University, Leiden, The Netherlands \\
%\email{\{s.o.brand, t.h.w.baeck, a.w.laarman\}@liacs.leidenuniv.nl}
\email{s.o.brand@liacs.leidenuniv.nl \\ t.h.w.baeck@liacs.leidenuniv.nl \\ a.w.laarman@liacs.leidenuniv.nl}
}}
{
\author{Anonymized for FM 2023}
\authorrunning{Anonymized for FM 2023}
\institute{}
\email{}
}
\maketitle

% FM artefact badges
\SetWatermarkAngle{0}
\SetWatermarkText{\raisebox{12.5cm}{%
  \hspace{0.1cm}%
  \href{https://doi.org/10.5281/zenodo.7333633}{\includegraphics{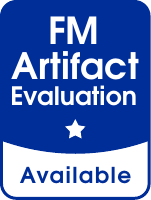}}% Update DOI to reflect your artefact
  \hspace{9cm}%
  \includegraphics{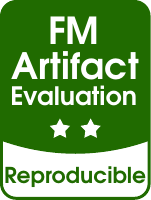}%
}}

% Abstract
\begin{abstract}
Saturation is considered the state-of-the-art method for computing fixpoints with decision diagrams. We present a relatively simple decision diagram operation called \reach that also computes fixpoints. In contrast to saturation, it does not require a partitioning of the transition relation. We give sequential algorithms implementing the new operation for both binary and multi-valued decision diagrams,
 and moreover provide parallel counterparts.
We implement these algorithms and experimentally compare their performance against saturation on 692 model checking benchmarks in different languages. The results show that the \reach operation often outperforms saturation, especially on transition relations with low locality. 
In a comparison between parallelized versions of \reach and saturation we find that \reach obtains comparable speedups up to 16 cores, although falls behind saturation at 64 cores.
Finally, in a comparison with the state-of-the-art model checking tool ITS-tools we find that \reach outperforms ITS-tools on 29\% of models, suggesting that \reach can be useful as a complementary method in an ensemble tool.

\end{abstract}

% Sections
\section{Introduction}

%\vspace{-.7em}
% motivation + background (model checking, reachability analysis, decision diagrams)
\paragraph{Reachability analysis}
Model checking is an important technique for ensuring that systems work according to specification. A core task in model checking is reachability analysis~\cite{biere2002liveness,cook2006terminator}, 
i.e., computing forward or backward reachable states of a system.
%The verification of many different properties of transition systems, such as liveness properties, can be reduced a reachability problem \cite{biere2002liveness,cook2006terminator}. 
Typically, the state space of a program grows exponentially with the number of variables and threads.
One method for dealing with this explosion is the use of symbolic methods such as decision diagrams.
%Many types of decision diagrams exist, such as binary decision diagrams (BDDs) \cite{bryant1986graph}, algebraic decision diagrams (ADDs) \cite{bahar1997algebric},  multi-valued decision diagrams (MDDs) \cite{kam1998multi} and sentential decision diagrams (SDDs)~\cite{darwiche2011sdd}.%\todo{fix bibliography: capital abbr.s and order (`van de Pol' indexes to Pol, vd)}
Decision diagrams\cite{bryant1986graph} are directed, acyclic graphs that succinctly represent sets of states  by leveraging the exponential growth in paths from a dedicated root node to a leaf.
%, required for reachability analysis, can implemented on decision diagrams as well, and although theoretically NP-hard \cite{mcmillan1992symbolic}, this operation is often efficient in practice.
The data structure provides various efficient manipulation operations, such as logical disjunction and conjunction and image computation.

In this work, we present a new decision diagram operation for reachability.
% and empirically evaluate this against a number of other common approaches, most notably saturation \cite{ciardo2001saturation} which is generally considered the most efficient approach for computing reachability with decision diagrams \cite{satunbound}.

% related work
\paragraph{Related work}
While SAT-based methods for model checking have become increasingly popular, doing reachability analysis with decision diagrams is still an important component of many state-of-the-art model checking tools, as can be seen in the Model Checking Contest (MCC) \cite{mcc2021,thierry2015symbolic}. %, leading tools employ hybrid methods which include the use of decision diagrams.
Symbolic reachability analysis with binary decision diagrams (BDDs) and other variants \cite{bahar1997algebric,kam1998multi,darwiche2011sdd,vinkhuijzen2020symbolic} is done by encoding both the initial system state $S_{\text{init}}$ and its transition relation $R$ in the diagram.
The set of reachable states $S$ is then iteratively computed using the image operation~\cite{mcmillan1992symbolic}, denoted by $S.R$, starting from $S_{\text{init}}$.
%he image of $S$ under $R$
%can be implemented as a decision diagram operation using conjunction and existential quantification. % \cite{bryant1992symbolic}
Since the order of exploration (e.g. breadth-first search, depth-first search, or other strategies) greatly influences the sizes of the intermediate decision diagrams, various exploration strategies, like saturation \cite{ciardo2001saturation}, chaining \cite{roig1995verification}, and sweep-line \cite{christensen2001sweep},
have been considered.
These algorithms have in common the use of the image computation $S.R$ as their main operation.%, either using the complete relation $R$, or a number of partial relations $R_1, R_2, \dots, R_k$.

Saturation stands out from other approaches, not only because it often performs better~\cite{satunbound} and leading MCC tools use it in decision diagram based reachability  \cite{thierry2015symbolic}, but also because it integrates the image computations into the traversal of the decision diagram of $S$.
%Decision diagrams consist of a number of ``levels,'' one for each variable in the analyzed system.
Saturation avoids redundant reconstructions by building the decision diagram for the reachable states bottom-up, eagerly \textit{saturating} the bottom nodes by exhaustively applying all relevant transitions.

%A different approach is the work by Matsunaga et al. \cite{matsunaga1993computing}, where a BDD operation is presented for directly computing the reflexive-transitive closure $R^*$, from which the set of reachable states $S.R^*$ can be easily computed. However, as confirmed by the authors, computing $R^*$ is generally much more expensive in both space and time requirements than computing $S.R^*$. 
% which makes it poorly suited as a reachability algorithm.

% contribution
\paragraph{Contribution}
We present three new decision diagram operations for reachability: \reachbdd and \reachmdd (for BDDs and MDDs respectively), as well as a parallel version of \reachbdd. % These algorithms 
%are inspired by an earlier relational transitive closure operation. But
%avoid the expensive computation of the closure $R^*$, computing only the reachable states $S.R^*$.
These algorithms partially construct the decision diagram of $S$ from the bottom up, but 
unlike saturation do not require partial relations and can handle monolithic transition relations. %\todo{we also show how partial relations can be supported.}
An additional advantage of these new reachability operations is their relative simplicity in comparison with saturation.

% FPT claim
%We show that image computation, as well as reachability operations, are fixed parameter tractable~\cite{downey2012parameterized} in the bandwidth $2^k$ of the transition relation matrix, meaning that the operations are tractable for fixed $k$. This is the first FPT result for these operations, as far as we can tell, and perhaps surprising, as graph reachability with BDDs as inputs is PSPACE-hard~\cite{feigenbaum1998complexity}  (image computation is \NP-complete \cite{mcmillan1992symbolic}).

We implement these new operations in the decision diagram package Sylvan~\cite{vandijk2013multi}, and experimentally compare them against saturation on a total of 692 problem instances from three model checking benchmark sets:  DVE (BEEM \cite{pelanek2007beem}), % 300 dve
Petri nets (MCC \cite{mcc2016}), % 357 petri nets?
and Promela models~\cite{holzmann1997model,van2013spins}.
We find that our methods are competitive with saturation, and tend to outperform saturation on larger instances. 
The parallel speedups obtained by \reachbdd up to 16 cores are comparable to those achieved in a parallel version of saturation~\cite{vandijk2019saturation}, although they fall behind on 64 cores.
Aside from the comparison against saturation, we also compare \reachmdd against the state-of-the-art model checking tool ITS-tools~\cite{thierry2015symbolic}, where we find that \reachmdd performs better than ITS-tools on 29\% of models, and can therefore be useful as a complementary method in an ensemble tool.
%\todo{verwerk overview's forward references in intro}
%\todo{Brag about new MCC results. Also in abstract. Quantify.}

% ,vandijk2015sylvan,vandijk2017sylvan
% the benchmark set compiled in \cite{oortwijn2017distributed}

% section overview (we can keep this outline in as long as it's not an issue for the page limit)
\paragraph{Outline}
\autoref{sec:prelims} discusses preliminaries. \autoref{sec:algorithm} explains the new reachability algorithms, and is then followed by an empirical evaluation of these algorithms in \autoref{sec:experiments}. Finally, \autoref{sec:conclusion} concludes this work.

\section{Preliminaries}
\label{sec:prelims}
%In this section, we explain decision diagrams and how they are used to compute reachability. We also introduce the notation used in the rest of the paper.
%
%
%\subsection{Decision Diagrams}
%\label{sec:prelims-dds}

%\vspace{-.5em}
\subsection{Binary decision diagrams}\label{sec:bdd}
Binary decision diagrams (BDDs) \cite{bryant1986graph,somenzi1999binary} are a data structure for representing Boolean functions $f(x_1,\dots, x_n)$, i.e., functions of type
$f : \{0, 1\}^n \to \{0, 1\}$.
Structurally, a BDD is a rooted, directed, acyclic graph with two types of nodes: terminal nodes with values $\{0, 1\}$ and non-terminal nodes $v$ that have two children, $\bddlow{v}$ (\concept{low}) and $\bddhigh{v}$ (\concept{high}), and a variable label $\bddvar{v} \in \{1, \dots, n\} = [n]$, indexing into $\set{x_i}_{i\in[n]}$.
%For the leafs, we fix $\bddvar{0} = \bddvar{1} = 0$.
\autoref{fig:example-bdd} shows examples of \concept{ordered}
BDDs, i.e., BDDs where on each path %from root  to leaf 
variable labels occur in a fixed order $x_1 < x_2 < \dots < x_n$.

\begin{figure}[b!]
    \centering \vspace{-2em}
%    \subfloat[] {\scalebox{.8}{\input{figures/example-bdd-0}}} \hspace{4em}
    \subfloat[\label{sub:a}] {\scalebox{.8}{\begin{tikzpicture}[auto, thick,node distance=1.cm,inner sep=1.5pt] 
	\node[] (0) [] {};
	\node[main] (1) [node distance=.6cm,below of=0] {$x_1$}; % \node[main,,label={right:$r$}]
	\node[main] (2) [below left of=1] {$x_2$};
	\node[main] (3) [below right of=1] {$x_2$};
	\node[main] (4) [below of=2] {$x_3$};
	\node[main] (5) [below of=3] {$x_3$};
	\node[leaf] (6) [below of=4] {0};
	\node[leaf] (7) [below of=5] {1};
	\draw[->] (0) to (1);
	\draw[->] (1) to (2);
	\draw[->,dashed] (1) to (3);
	\draw[->,dashed,bend right] (2) to (6);
	\draw[->] (2) to (4);
	\draw[->] (4) to (6);
	\draw[->,dashed] (4) to (7);
	\draw[->,bend left] (3) to (7);
	\draw[->,dashed] (3) to (5);
	\draw[->] (5) to (6);
	\draw[->,dashed] (5) to (7);
\end{tikzpicture}}} \hspace{4em}
    \subfloat[\label{sub:b}] {\scalebox{.8}{\begin{tikzpicture}[auto, thick,node distance=1cm,inner sep=1.5pt] 
	\node[] (0) [] {};
	\node[main] (1) [node distance=.6cm,below of=0] {$x_1$}; % \node[main,,label={right:$r$}]
	\node[main] (2) [below left of=1] {$x_2$};
	\node[main] (3) [below right of=1] {$x_2$};
	\node[main] (4) [below of=2] {$x_3$};
	\node[main] (5) [below of=3] {$x_3$};
	\node[leaf] (6) [below of=4] {0};
	\node[leaf] (7) [below of=5] {1};
	\draw[->] (0) to (1);
	\draw[->] (1) to (2);
	\draw[->,dashed] (1) to (3);
	\draw[->,dashed,bend right] (2) to (6);
	\draw[->] (2) to (4);
	\draw[->] (4) to (6);
	\draw[->,dashed] (4) to (7);
	\draw[->] (3) to (5);
	\draw[->,bend right] (5) to (7);
	\draw[->,dashed,bend left] (5) to (7);
	\draw[->,dashed] (3) to (4);
\end{tikzpicture}}} \hspace{4em}
    \subfloat[\label{sub:c}] {\scalebox{.8}{\begin{tikzpicture}[auto, thick,node distance=1cm,inner sep=1.5pt] 
	\node[] (0) [] {};
	\node[main] (1) [node distance=.6cm,below of=0] {$x_1$}; % \node[main,,label={right:$r$}]
	\node[main] (2) [below left of=1] {$x_2$};
	\node[main] (3) [below right of=1] {$x_2$};
	\node[main] (4) [below of=2] {$x_3$};
	\node[] (5) [below of=3] {};
	\node[leaf] (6) [below of=4] {0};
	\node[leaf] (7) [below of=5] {1};
	\draw[->] (0) to (1);
	\draw[->] (1) to (2);
	\draw[->,dashed] (1) to (3);
	\draw[->,dashed,bend right] (2) to (6);
	\draw[->] (2) to (4);
	\draw[->] (4) to (6);
	\draw[->,dashed] (4) to (7);
	\draw[->] (3) to (7);
	\draw[->,dashed] (3) to (4);
\end{tikzpicture}}}
    \caption{(Ordered) BDDs representing function $f = (x_1 \land x_2 \land  \overline{x_3}) \lor (\overline{x_1} \land \overline{x_2} \land \overline{x_3}) \lor (\overline{x_1} \land x_2)$.
For node $v$, we draw $\bddvar{v}=i$ as $x_i$. Dashed lines lines represent low branches
($\bddlow{v}$) and solid lines high branches ($\bddhigh{v}$).
Only the BDD in (c) is completely reduced. The BDD in (a) can be reduced to (c) by merging the two isomorphic nodes for $x_3$, while the BDD in (b) can be reduced to (c) by removing the (right-most) redundant node~$x_3$.}
    \label{fig:example-bdd}
    \centering \vspace{-1em}
\end{figure}

\begin{wrapfigure}[3]{r}{5em}
%\begin{figure}[H]
\vspace{-3.0em}
    \centering
    \scalebox{.9}{\begin{tikzpicture}[node distance=10mm, inner sep=1.5pt, auto, thick] 
	\node[] (0) [] {};
	\node[main,label=right:{$v$}] (1) [node distance=7mm,below of=0] {$x_i$};
	\node[] (2) [below left of=1] {$v[0]$};
	\node[] (3) [below right of=1] {$v[1]$};
	\draw[->] (0) to node [] {$f^v$} (1) ;
	\draw[->,dashed] (1) to (2);
	\draw[->] (1) to (3);
\end{tikzpicture}}
\vspace{-1.5em}
%    \caption{BDD node.}
%    \label{fig:shannon-bdd}
%\end{figure}
\end{wrapfigure}
A non-terminal node $v$ with $\bddvar{v} = i$, shown right,
can be read as the \concept{Shannon decomposition} \textit{``if $x_i = 1$ then $\bddhigh{v}$ else $\bddlow{v}$.''}
%or shorter: $v[x_{\bddvar{v}}]$.
The function represented by the node $v$, call it $f^v$, is thus given by
%The node $v$ thus represents the function~\cite{10.1007/978-3-030-17465-1_17}:
% $f^v$ as follows:
\begin{align}
f^v(x_1, \dots, x_n) \defn
 \begin{cases}
    x_i f^{v[1]} \lor \overline{x_i} f^{v[0]} & \text{if $v \notin \set{0,1}$,} \\
    v & \text{if $v \in \set{0,1}$.} %\text{if $v$ is a terminal node $\{0,1\}$}  \\
    %f^v[x_{\bddvar{v}}]  &  \text{if $v\notin \set{0,1}$}  \\
 \end{cases}%\nonumber
 \label{eq:sub-function-dd}
\end{align}

\noindent
The definition in \autoref{eq:sub-function-dd} shows that a conditioned subfunction $\con {f^v}{\vec a}$, as in \autoref{eq:sub-function} below,
is represented by a decision diagram node, namely
the one following the path $v[\vec a] \defn v[a_1][a_2]\dots[a_k]$, assuming the BDD is ordered and no variables are skipped.
%\todo{what about $f[1][1]$ where $f = 1$. Think of  a neat way to extend formalism a bit. Careful, this paragraph is now one long sentence.}
\begin{align}
        f^v_{|a_1, \dots, a_k}(x_1,\dots, x_n) 
        %\defn f|_{x_1=a_1, \dots, x_k=a_k} 
        \defn f^v(x_1 = a_1, \dots, x_k = a_k, x_{k+1}, \dots, x_n)
    \label{eq:sub-function}
\end{align}

The insight that BDDs exploit to realize succinct representations of commonly-encountered Boolean functions is that  many  subfunctions for different $\vec a,\vec b \in \set{0,1}^*$ can be isotropic, i.e., $\con f{\vec a} = \con f{\vec b}$. Take $\con f{00} = \con f{11}$  in \autoref{sub:a}.
In the diagram, this means that the isomorphic subgraphs below the nodes representing $\con f{00}$ and $\con f{11}$ can
be merged, as in \autoref{sub:c}.

An ordered BDD is \concept{reduced} when, in addition to isomorphic sub-graph 
merging, all \concept{redundant nodes} (nodes $v$ with $v[0]=v[1]$) are removed.
Take $\con f{010} = \con f{011} = \con f{01} = 1$, in \autoref{sub:c}.  
Removed redundant nodes can be reconstructed by recognizing that a variable is skipped on a path, as done \autoref{sub:b}.

Reduced and ordered BDDs (ROBDDs) are canonical representations of Boolean functions, i.e., any two functions with the same truth table are uniquely  represented by (the root node of) an ROBDD.
Canonicity allows for equivalence checking in constant time through hashing of nodes $v$ as tuples \tuple{\bddvar v, \bddhigh{v},\bddlow{v}}
	(provided that nodes \bddhigh{v},\bddlow{v} are already canonically represented, i.e., the BDD is build in a bottom-up fashion). 
%provided the BDD is constructed in a bottom-up fashion
This in turn allows efficient manipulation operations ($\land, \lor, \dots$) as discussed~below.
%Moreover, because all isomorphic sub graphs are merged, the BDD is reduced in size.

In this text, we fix the variable order  $x_1 < x_2 < \dots < x_n$.
For conciseness, we will often consider quasi-ROBDD, which are ROBDDs where
all redundant nodes are reconstructed (e.g. \autoref{sub:b}). In many settings,
this stronger definition does not lose generality, as quasi-ROBDD are also canonical,
and at most a factor $\frac n2$ larger than an ROBDD for the same function~\cite{knuth2011taocp4a}.
%\todo{kan korter. wordt nu minder in bewijzen gebruikt. Wel in 2.3}~
At the same time, because for all quasi-ROBDD nodes $v$ we have $f^v_{\vec a} = v[\vec a]$,
this assumption greatly simplifies algorithm representation and reduces cases in proofs.
We denote with \concept{level} $i$  the set of all nodes with the  variable label $i$.

\subsection{Multi-valued decision diagrams} 
\label{sec:mdds}
Multi-valued decision diagrams (MDDs) \cite{kam1998multi,sanner2005affine} are a generalization of BDDs for encoding functions $\mathcal{D}_1 \times  \dots \times \mathcal{D}_n \to \{0, 1\}$, where 
 $\mathcal{D}_i=\{0,1,\dots,m-1\}$ for some $m$ and all $i$.
%$n$ is the number of variables and $\mathcal{D_i}$ is typically a finite subset of the natural numbers.
%These functions can be interpreted as characterizing functions of sets of $n$-tuples of integers.
Each MDD node with variable $x_i$ has $m$ outgoing edges, each with a label in $\mathcal{D}_i$. The interpretation of following an edge remains the same as for BDDs: for an MDD which encodes a function $f$ and has root node $v$ with $\bddvar{v}=i$, following an edge with label $a \in \mathcal{D}_i$ leads to an MDD which encodes the sub-function $f|_{x_i = a}$.

\begin{figure}[b]
    \centering
    \vspace{-1em}
    \subfloat {\scalebox{.8}{\begin{tikzpicture}[node distance=3mm and 3mm, thick, edgelabel/.style={fill=white,rounded corners=4pt,inner sep=1.5pt}]
    
	\node[] (root) [] {};
	\node[main] (r0) [node distance=3mm,below=of root,inner sep=1.5pt] {$x_1$};
	
	\begin{scope}[node distance=6mm and 7mm,inner sep=1.5pt]
	\node[main,inner sep=1.5pt] (w1) [below=of r0,,label=left:{$B$}] {$x_2$};
	\node[main,inner sep=1.5pt] (w0) [left=of w1,label=left:{$A$}] {$x_2$};
	\node[main,inner sep=1.5pt] (w2) [right=of w1] {$x_2$};
	\node[leaf] (t0) [below=of w0] {1};
	\node[leaf] (t1) [below=of w1] {1};
	\node[leaf] (t2) [below=of w2] {1};
	
	\draw[->] (root) to (r0);
    
	\draw[->] [out=-170,in=90] (r0) to node[edgelabel] {0} (w0);
	\draw[->] [out=-125,in=45] (r0) to node[edgelabel] {2} (w0);
	\draw[->] (r0) to node[edgelabel] {5} (w1);
	\draw[->] [out=-15,in=105] (r0) to node[edgelabel] {7} (w2);
	
	\draw[->] [out=-135,in=135] (w0) to node[edgelabel] {0} (t0);
	\draw[->] [out=-90,in=90] (w0) to node[edgelabel] {1} (t0);
	\draw[->] [out=-45,in=45] (w0) to node[edgelabel] {2} (t0);
	
	\draw[->] [out=-120,in=120] (w1) to node[edgelabel] {1} (t1);
	\draw[->] [out=-60,in=60] (w1) to node[edgelabel] {2} (t1);
	
	\draw[->] [out=-120,in=120] (w2) to node[edgelabel] {0} (t2);
	\draw[->] [out=-60,in=60] (w2) to node[edgelabel] {2} (t2);
	\end{scope}
\end{tikzpicture}}} \hspace{4em}
    \subfloat {\scalebox{.8}{\begin{tikzpicture}[node distance=6mm and -0.3mm, thick]
	\node[] (root) [] {};
	\node[lddnode] (r0) [node distance=3mm,below=of root,label=left:{$x_1:$}] {0};
	\node[lddnode] (r1) [right=of r0] {2};
	\node[lddnode] (r2) [right=of r1] {5};
	\node[lddnode] (r3) [right=of r2] {7};
	\node[lddnode] (w0) [below=of r0,label=left:{$x_2:$}] {0};
	\node[lddnode] (w1) [right=of w0] {1};
	\node[lddnode] (w2) [right=of w1] {2};
	\node[] (skip1) [right=of w2] {};
	\node[] (skip2) [right=of skip1] {};
	\node[lddnode] (w3) [right=of skip2] {0};
	\node[lddnode] (w4) [right=of w3] {2};
	\node[leaf] (terminal) [below=of w2] {1};
	
	\draw[->] (root) to (r0);
	
	\draw[->] (r0) to (w0);
	\draw[->] (r1) to [out=-90,in=75] (w0);
	\draw[->] (r2) to [out=-90,in=75] (w1);
	\draw[->] (r3) to [out=-90,in=105] (w3);
	
	\draw[->] (w0) to [out=-90,in=155] (terminal);
	\draw[->] (w1) to [out=-90,in=125] (terminal);
	\draw[->] (w2) to [out=-90,in=90] (terminal);
	\draw[->] (w3) to [out=-90,in=55] (terminal);
	\draw[->] (w4) to [out=-90,in=25] (terminal);
\end{tikzpicture}}}
    \caption{An MDD (left) and LDD (right) which both encode the set 
    $\{ \langle 0,0 \rangle, \allowbreak
        \langle 0,1 \rangle, \allowbreak
        \langle 0,2 \rangle, \allowbreak
        \langle 2,0 \rangle, \allowbreak
        \langle 2,1 \rangle, \allowbreak
        \langle 2,2 \rangle, \allowbreak
        \langle 5,1 \rangle, \allowbreak 
        \langle 5,2 \rangle, \allowbreak
        \langle 7,0 \rangle, \allowbreak
        \langle 7,2 \rangle \}$.
        Arrays represent right-ward sibling chains in the LDD.
        To improve legibility, 
%        for the MDD three terminal nodes have been drawn instead of one, and for both figures 
        we omit edges pointing to the 0 terminal and replicate the 1 terminal for MDDs.}
    \label{fig:mdd-ldd-example}
\end{figure}

Similar to BDDs, MDDs are typically reduced by merging isomorphic sub-graphs.
However, unlike ROBDDs, redundant nodes are usually not removed in MDDs, which means variables are never skipped on any path. 
So an MDD with $m=2$ is a Quasi-ROBDD.
%Despite this so-called ``quasi-reduction'' MDDs are canonical representations of functions,
%like BDDs.
\autoref{fig:mdd-ldd-example} shows an example.

A list decision diagram (LDD) \cite{blom2008symbolic} is the Knuth transform of an MDD into a left-child right-sibling binary tree. 
Siblings (right-ward chains) are stored as (ordered) linked lists, which allows the reuse of common sibling suffixes, as shown in the example in \autoref{fig:mdd-ldd-example} (e.g. $f_{|5}$ reuses a part of the siblings of $f_{|2} = f_{|0}$).
%: the nodes labeled $A$ and $B$ have outgoing edges $\{0,1,2\}$ and $\{1,2\}$ respectively.

Finally, a BDD or MDD representing a function $f(x_1, \dots, x_n)$ can also be interpreted as a set of strings $\vec a$ of length $n$, according to the characteristic function \set{\vec a \mid f(\vec a) = 1}. So the BDDs in \autoref{fig:example-bdd} all represent $\{000, 010, 011, 110\}$.

% Operations
\subsection{Decision diagram operations}

What makes BDDs and MDDs so useful, aside from their possible succinctness, is that many 
manipulation operations, such as disjunction (set union) and conjunction (set intersection), can be performed in polynomial time in the number of decision diagram nodes in the operands~\cite{bryant1986graph,amilhastre2014compiling}.
While other operations  that have been shown to be \NP-complete~\cite{mcmillan1992symbolic}, such as unbounded existential quantification and thus also image computation~\cite{mcmillan1992symbolic}, 
	are often still efficient in practice~\cite{burch1992symbolic}.
%\todo{move this to elsewhere? Conclusion? Where we need it in discussion?}

{
\setlength{\interspacetitleruled}{0pt}
\setlength{\algotitleheightrule}{0pt}
\begin{algorithm}[b]
% (algorithm to put into a algorithm2e algorithm environment)
\SetKwFunction{FApply}{Union}
\Fn{\FApply{$A, B$}}{
     \vspace{-1.2em}\Comment*[r]{For quasi-ROBDDs $A$ and $B$ on $n$ variables.}
     \BlankLine
    \lIf{$A = 0$}{\Return{$B$}}
    \lIf{$B = 0$}{\Return{$A$}}
	\lIf{$A = 1 \lor B = 1$}{\Return{1}}
	\BlankLine
	\lIf{res $\gets$ cache[$\fname{Union},A,B$]}{\Return{res}} \label{line:cache-lookup}
	\BlankLine

   $x \gets  \bddvar A$% \min(\bddvar A, \bddvar B )$	\; \label{line:min}
   \Comment*[r]{By virtue of quasi reduction \bddvar{$A$} = \bddvar{$B$}.}   
%   	\textbf{assert} \bddvar{$A$} = \bddvar{$B$} \Comment*[r]{By virtue of quasi reduction.}    
	$L \gets$ \FApply{$\e A0, \e B0$} \; \label{line:call-low} 
	$H \gets$ \FApply{$\e A1, \e B1$} \scalebox{.95}{\begin{tikzpicture}[node distance=10mm, inner sep=1.5pt, auto, thick, overlay, xshift=1.8cm] 
	\node[] (0a) [] {};
	\node[main] (1a) [node distance=7mm,below of=0a] {$x$};
	\node[] (2a) [below left  = 4mm and -1.5mm of 1a] {$A[0]$};
	\node[] (3a) [below right = 4mm and -1.5mm of 1a] {$A[1]$};
	\draw[->] (0a) to node [] {$A$} (1a) ;
	\draw[->,dashed] (1a) to (2a);
	\draw[->] (1a) to (3a);
	
	\node[] (union) [right = .4cm of 1a] {$\union$};
	
	\node[main] (1b) [right = .4cm of union] {$x$};
	\node[] (0b) [node distance=7mm, above of=1b] {};
	\node[] (2b) [below left  = 4mm and -1.5mm of 1b] {$B[0]$};
	\node[] (3b) [below right = 4mm and -1.5mm of 1b] {$B[1]$};
	\draw[->] (0b) to node [] {$B$} (1b) ;
	\draw[->,dashed] (1b) to (2b);
	\draw[->] (1b) to (3b);
	
	\node[] (equals) [right = .8cm of 1b] {$=$};
	
	\node[main] (1c) [right = 1.2cm of equals] {$x$};
	\node[] (0c) [node distance=7mm, above of=1c] {};
	\node[] (2c) [below left  = 4mm and 0mm of 1c] {$A[0] \union B[0]$};
	\node[] (3c) [below right = 4mm and 0mm of 1c] {$A[1] \union B[1]$};
	\draw[->] (0c) to node [] {res} (1c) ;
	\draw[->,dashed] (1c) to (2c);
	\draw[->] (1c) to (3c);
\end{tikzpicture}} \; \label{line:call-high} 
	res $\gets$ \fname{MakeNode}($x, L, H$)\; 
	\BlankLine
	cache[$\fname{Union},A,B$] $\gets$ res\; \label{line:cache-put}
	\BlankLine
	\Return{res}
}

\end{algorithm}
}

As an example, the algorithm below computes the union of two quasi-reduced
diagrams
%\todo{check elsewhere!}
$A$ and $B$,
i.e., $A \cup B$ (or in the functional interpretation: $A \lor B$).
Because any BDD is defined by its root node, the arguments $A$ and $B$ are simply given as nodes.
%instead of $\bddvar{\rootnode{A}}$ we might equivalently write $\bddvar{A}$.
The algorithm first considers leafs as a base case, treating them according to the semantics of $\lor$.
On \autoref{line:call-low} and \ref{line:call-high}, the function is called recursively on the children of the input BDDs, synchronizing on the low and high branches. The results from these recursive calls are then combined with the $\fname{MakeNode}(x, L, H)$ function which creates a \emph{reduced BDD node} $v$ with $v[0]= L$, $v[1] = H$ and $\bddvar v=x$.
To ensure reduction, it returns $L$ ($=H$) when the node is redundant and looks up the tuple $\tuple{x, L, H}$ in a \concept{unique table}, as discussed in \autoref{sec:bdd}.
For MDDs, we assume a function $\fname{MakeNode}(x, A_0, \dots, A_{m-1})$ which returns a (quasi-)reduced MDD node and, for notational convenience, takes $m+1$ positional arguments: a variable $x$, and one MDD node $A_i$ for each of its $m$ children (some of which can be 0).
%, which encodes the function $\overline{x_i} \ell \lor x_i h$.
%On ~\autoref{line:min}, the minimum variable is taken of both nodes
%(note that $A_{|x=1} = A_{|x=0} = A$ when $x$ is skipped in $A$).

Lastly, it is important to realize that a decision diagram with $|V|$ nodes can have $\exp(|V|)$ paths from the root to a terminal node. To achieve polynomial runtimes,
decision diagram operations use top-down dynamic programming (see \autoref{line:cache-lookup},~\ref{line:cache-put}).
	This ensures that different paths leading to the same (pairs of) nodes are caught by the cache, avoiding recomputation.
%but with dynamic programming the algorithm becomes polynomial-time 
% (see \autoref{line:cache-lookup},~\ref{line:cache-put}), since different paths leading to the same pairs of nodes are caught by the cache, and do not need to be recomputed.

%consult a separate cache table for this purpose.

%Reachability analysis with BDDs will take, in the worst case, exponential resources.
%It is know that image computation is NP-hard \cite{mcmillan1992symbolic}, and even more so: reachability with BDDs is proven to be outside of PSPACE \cite{bollig2010exponential}.
% ... despite proven to be efficient in practice for many important problems

\subsection{Encoding symbolic transition systems}
\label{sec:symbolic-ts}
% = \set{0,1}^{\set{x1,\dots,x_n}}
A symbolic transition system is a tuple $(\vec x, \sspace, R, S_{\text{init}})$, where
$\vec x$ is a tuple of Boolean variables $\vect{x_1,\dots,x_n}$, 
$\sspace = \set{0,1}^{n}$ is the state space (including unreachable states), $R \subseteq \sspace \times \sspace'$ is a transition relation, and $S_{\text{init}} \subseteq \sspace$ is a set of initial states. The relation $R$ is a constraint over variables $\vec x, \vec x'$, where $\vec x$ encodes the source states and $\vec x'$ (consisting of primed copies of $\vec x$) encodes target states.
While $R$ monolithically encodes the system's behavior, we also consider a local variant, discussed in the example which follows.

As an example we give a transition system which captures the dining philosophers problem. We have $k$ processes (philosophers) and $k$ resources (forks). Each process $P_i$, with $i \in \{1,\dots,k\}$, attempts to allocate two resources: a fork $i$ on the left and a fork $j = ((i-1) \bmod k) + 1$ on the right. If a fork is unavailable the philosopher waits until it becomes available.
To describe the state space we use $3k$ Boolean variables: $\vec x = (a_1, l_1, r_1, \dots, a_{k}, l_{k}, r_{k})$, where $a_i$ $(\overline{a_i})$ indicates fork $i$ is (not) available, $l_i$  $(\overline{l_i})$ indicates philosopher $i$ does (not) hold a fork in their left hand, and $r_i$  $(\overline{r_i})$ idem for their right hand.
The starting state $S_{\text{init}} = (a_1, \overline{l_1}, \overline{r_1}, \dots, a_{k}, \overline{l_{k}}, \overline{r_{k}}$.)

To define $R_i$ for $k > 1$ processes, we define three local relations $R_i^{m}$ that implement picking up and putting down the left/right fork, and eating.
\begin{align*}
    R_i^1 &= 
    (a_i \oplus a_i') \land ({l_i} \oplus l_i')
%    \lor(\overline{a_i} \land a_i' \land {l_i} \land \overline{l_i'})
    & & \triangleright \texttt{ pick up / put down left} \\
    R_i^2 &= 
    ({r_i} \oplus r_i') \land (a_j \oplus {a_j'})
%    \lor(r_i \land \overline{r_i'} \land \overline{a_j} \land a_j')
    & & \triangleright \texttt{ pick up / put down right} \\
    R_i^3 &= 
    l_i \land l_i' \land r_i \land r_i'
    & & \triangleright \texttt{ eat (hold on to both forks)}
\end{align*}
Let $\svars{R_i^m}$ be the set of (primed and unprimed) variables in $R_i^m$.
Notice that the support of each local relation contains a different subset of the target variables $\vec x'$. To 
ensure that the other target variables are not left unconstrained,
%combine them into a single $R_i$ by means of a union/logical or, 
we need to add the constraint $x \iff x'$ for all missing variables  $x' \in \vec x'$. For instance, pick up / put down left should be extended as:
\begin{align*}
    &  (a_i \oplus a_i') \land ({l_i} \oplus l_i') \land (r_i \iff r_i') \land \Land_{j\neq i}  (a_j \iff a_j' \land
    r_j \iff r_j' \land l_j \iff l_j')
    & & %\triangleright \texttt{ pick/put left, leave right}
\end{align*}
The same needs to happen when merging multiple processes $R_i$ into a single global relation $R$, 
i.e. $R = \Lor_{i=1}^{k} \mathcal{R}_i$, where $\mathcal{R}_i \defn R_i \land  \Land_{x' \in \vec x' \setminus \svars{R_i}} x \iff x'$.
\autoref{sec:prelim-reach} discusses how extending local relations in this way can be avoided.

If a transition relation $R$ is composed of $R_i$'s, where each $\svars{R_i}$ is a small subset of $\vec x, \vec x'$, we say that $R$ has high locality. If $R$ cannot be split up into such partial relations we say that $R$ has low locality, or is \textit{monolithic}.

\subsection{Reachability with decision diagrams}
\label{sec:prelim-reach}
% Encoding S
For a transition system with $n$ Boolean variables $\{x_1, \dots, x_n\}$, a single state is given by a bit string $s \in \{0,1\}^n$. A set of states $S \subseteq \{0,1\}^n$ can be encoded in a BDD. %, as discussed in \autoref{sec:prelims-dds}.
Likewise, a transition relation $R \subseteq \{0,1\}^n \times \{0,1\}^n$ can be encoded in a BDD with $2n$ variables.
The variables are ordered by interleaving the source state variables $\vec x = (x_1,\dots, x_n)$ with target state variables (primed copies: $x_1',\dots, x_n'$), i.e., ($x_1, x_1', \dots, x_n, x_n'$), as it is both convenient for the implementation of BDD algorithms, and it reduces the diagram size.
%
%% reachability with BDDs
%Once we have two BDDs $S$ and $R$, a 
%
%decision diagram operation $\fname{Image}(S,R)$ can be created which returns the successor states. This operation consists out of a combination of other BDD operations like conjunction and existential quantification 
To compute the successors to an initial set of states $S$, we can then use the
decision diagram operation $\fname{Image}(S,R)$ 
\cite{mcmillan1992symbolic,vandijk2013multi}:
\begin{align}
\fname{Image}(S,R) = (\exists x_1, \dots, x_n : (S \land R))_{[x_1', \dots, x_n' \,\,:=\,\, x_1, \dots, x_n]},
\label{eq:image}
\end{align}
%where $x$ are the source variables and $x'$ the target variables, 
%the successors are computed with $\exists x : (S \land R))$, and
where $[x' := x]$ indicates the relabeling of the target variables to source \mbox{variables}. 
%Computing $\fname{Image}(S,R)$ can also be done in single BDD operation rather than three separate ones \cite{vandijk2013multi}.
The \fname{Image} operation can also be implemented for partial relations $R_i$, with the benefit that existential quantification only needs to happen over $\svars{R_i}$, rather than over all variables. The union of the image under each $R_i$ separately equals the image under the global transition relation:
%\begin{align}
%\fname{ImagePartial}(S,R_i) = (\exists \svars{R_i} : (S \land \R_i))_{[\vec x' \,\,:=\,\, \vec x]},
%\end{align}
\begin{align*}
\exists \vec x : ((\mathcal{R}_1 \lor \dots \lor \mathcal{R}_k) \land S)_{[\vec x' \,\,:=\,\, \vec x]} 
	&= (\exists  \vec x_1' : R_1 \land S)_{[\vec x' \,\,:=\,\, \vec x]} \lor \dots \lor \\
	&\phantom{=.} (\exists \vec x_k' : R_k \land S)_{[\vec x' \,\,:=\,\, \vec x]}
\end{align*}
where $\vec x'_i = \svars{R_i} \cap \vec x'$ and  $\mathcal{R}_i$ is the extension of $R_i$ as defined in \autoref{sec:symbolic-ts}.

With \fname{Image}, the set of reachable states can be computed by repeatedly applying $R$ to a growing set of reachable states until no new states are found, which we denote with $S.R^*$.
%More formally, the reflexive-transitive closure $R^*$ can be defined as the union of the composition of $R$ with itself any number of times ($R^k = R\circ \dots \circ R$):
%\begin{align}
%    R^* \defequal \Union_{k=1}^\infty (R \union I)^k = \Union_{k=0}^\infty R^k,
%    \text{where $I = R^0$ is the identity relation.}
%    \label{eq:def-closure}
%\end{align}

Finally, decision diagrams often represent relations more succinctly when source variables
$\vec x$ are interleaved with target variables $\vec x'$ in the order~\cite{mcmillan1992symbolic}.

\subsection{Saturation}
\label{sec:prelim-saturation}
Saturation \cite{ciardo2001saturation} is a method for computing reachability that exploits locality of transitions. Aside from the initial states $S_{\text{init}}$, the algorithm takes as input several local transition relations $R_i$ (see \autoref{sec:symbolic-ts}), ordered such that $\bddvar{R_i} \geq \bddvar{R_{i+1}}$. 
%	i.e., $\min(\svars{R_i}) \geq \min(\svars{R_i})$.
%Assume we have $k$ such partial relations $R_1, R_2, \dots, R_k$, ordered in such a way that the variables in $R_1$ are at the bottom of the decision diagram, and the variables in $R_k$ are at the top.

To illustrate how saturation works, let us give an example. Say we have a global state space given by $\{x_1, \dots, x_4\}$, and three partial relations $R_i$ with $\svars{R_1} =$ $\{x_3, x_3', x_4, x_4'\}$, $\svars{R_2} =$ $\{x_2, x_2', x_3, x_3'\}$, and $\svars{R_3} =$  
%\begin{wrapfigure}[4]{r}{2.7cm}
%    \centering
%    \vspace{-1.0cm}
%    \scalebox{0.8}{\parbox{.5\linewidth}{%
%    \begin{align*}
%        \begin{blockarray}{ccccc}
%        & x_1 & x_2 & x_3 & x_4 \\
%        \begin{block}{r(cccc)}
%          R_1\phantom{/} & 0 & 0 & 1 & 1 \\
%          R_2\phantom{/} & 0 & 1 & 1 & 0 \\
%          R_3\phantom{/} & 1 & 1 & 0 & 0 \\
%        \end{block}
%        \end{blockarray} \\
%        %\begin{blockarray}{ccccc}
%        %& x_2 & x_4 & x_1 & x_3 \\
%        %\begin{block}{r(cccc)}
%        %  R_1\phantom{/} & 0 & 1 & 0 & 1 \\
%        %  R_2\phantom{/} & 1 & 0 & 0 & 1 \\
%        %  R_3\phantom{/} & 1 & 0 & 1 & 0 \\
%        %\end{block}
%        %\end{blockarray}
%    \end{align*}
%    }}
%\end{wrapfigure}
\begin{wrapfigure}[5]{r}{2.5cm}
    \centering
    \vspace{-0.95cm}
    \scalebox{0.8}{\parbox{.5\linewidth}{%
    \begin{align*}
        \begin{blockarray}{cccc}
        & R_1 & R_2 & R_3 \\
        \begin{block}{r(ccc)}
          x_1\phantom{/} & 0 & 0 & 1 \\
          x_2\phantom{/} & 0 & 1 & 1 \\
          x_3\phantom{/} & 1 & 1 & 0 \\
          x_4\phantom{/} & 1 & 0 & 0 \\
        \end{block}
        \end{blockarray} 
    \end{align*}
    }}
\end{wrapfigure}
$\{x_1, x_1', x_2, x_2'\}$.
The ``dependency matrix'' on the right visualizes the dependencies of the partial relations on each of the variables.
The saturation algorithm traverses the decision diagram of a set of states $S$ and \textit{saturates} the nodes from the bottom-up. What this means in the case of the example is that the partial relation $R_1$ is applied to all nodes $v$ in $S$ with $\bddvar{v} = 3$, until $v$ has converged.
After saturating this node $v$ the algorithm backtracks upwards in the decision diagram to a node $w$ with $\bddvar{w} = 2$, and saturates this node by exhaustively applying $R_2$, while eagerly saturating new nodes created below $w$.
%\alfons{Relate better to experimental section where you speak about locality.}
%\sebas{I've added something in experimental section to relate back to this section, but I am not sure if it makes it more clear/cohesive.}

While breadth-first search (BFS) suffers from large intermediate diagram sizes~\cite{valmari},
saturation often avoids this by ensuring that the lower levels of the decision diagram reach their final configuration early. Generally this works well if the (average) bandwidth (the distance between the first and the last non-zero entry in each row) of the dependency matrix is low. 
%\alfons{or can be made low through row/column permutations: cite Jeroen Meijer en ander related work.}
%\sebas{We already mention this at the end of this paragraph (although differently phrased), is this insufficient?} 
%Transition systems where this bandwidth is low are considered to have a lot of locality. 
This occurs for example in asynchronous systems where processes mainly modify local variables or communicate only with ``neighboring processes'' through channels or shared variables dedicated to neighboring pairs. This is for example the case in the dining philosophers example given in \autoref{sec:symbolic-ts}.
Finding a variable order and organizing the partial relations such that both this bandwidth and sizes of the decision diagrams are minimized is generally hard (even finding an optimal variable order for a single BDD is \NP-complete \cite{bollig1996improving}) but good heuristics
exist~\cite{aloul2001faster,aloul2003force,amparore2017gradient,amparore2018decision,meijer2016bandwidth}.

While originally proposed for MDDs, saturation has since been implemented for BDDs as well \cite{vandijk2019saturation}.

\section{Decision diagram operation for reachability}
\label{sec:algorithm}

\subsection{For BDDs}
\label{sec:alg-bdd}
We present an operation $\reachbdd(S,R)$ which computes the reachable states $S.R^*$ for BDDs. As is typical in BDD operations, our algorithm splits the computation into recursive calls on smaller, factored BDDs, after which the results are composed again in the backtracking step. The way we factor $R$ is inspired by \cite{matsunaga1993computing}, where a BDD algorithm is given for computing the closure $R^*$ of $R$ (computing the closure $R^*$ is generally much more expensive~\cite[\textsection6]{matsunaga1993computing}, 
hence we want to compute $S.R^*$ directly for a given $S$). \reachbdd is given in \autoref{alg:reach-rec}, and its correctness is discussed in \autoref{sec:correctness}\arxiv{ and \aref{app:reach-correctness-proof}}{}.

{
%\vspace{-0.5cm}
\begin{algorithm}[H]
% (algorithm to put into a algorithm2e algorithm environment)
\SetKwFunction{Reach}{ReachBdd}
\Fn{\Reach{$S$, $R$}}{
     \vspace{-1.2em}\Comment*[r]{For quasi-ROBDDs $S,R$ on $n,2n$ variables.}
	\BlankLine
	\lIf{$S = 0$}{\Return{0}}
	\lIf{$R = 0$}{\Return{$S$}}
	\lIf{$S = 1$}{\Return{1}}
	\lIf{$R = 1 \land S \neq 0$}{\Return{1}}
	\BlankLine \label{line:missing-cache-lookup}
	%\lIf{res $\gets$ cache[$S,R$]}{\Return{res}}
	%\BlankLine
	% $ S_0, S_1 \gets \e S{0}, \e S{1} $
	\While{$S$ did not converge}{ \label{line:loop}
		$\e S{0} \gets \reachbdd(\e {S}0, \e R{00})$ \; \label{line:loop1} \label{line:reach00}
		$\e S{1} \gets \fname{Union}(\e S{1}, \fname{Image}(\e S{0}, \e R{01}))$ \; \label{line:loop2} \label{line:img01}
		$\e S{1} \gets \reachbdd(\e S{1}, \e R{11})$ \; \label{line:reach11}
		$\e S{0} \gets \fname{Union}(\e S{0}, \fname{Image}(\e S{1}, \e R{10}))$ \;  \label{line:loop-end} \label{line:img10}
	}
	\label{line:missing-cache-put}
	%res = $\bar{x}_0S_0 \lor x_0 S_1$ \;
	%\BlankLine
	%cache[$S$,$R$] $\gets$ res \;
	%\BlankLine
	%\Return{res}
	\BlankLine
	\Return{$\makenode(\bddvar{S}, \e S{0}, \e S{1})$}
}

\caption{A BDD operation for computing reachability. Cache lookup/insert for dynamic programming after \autoref{line:missing-cache-lookup} and \ref{line:missing-cache-put} are omitted.}
\label{alg:reach-rec}
\end{algorithm}
}

The algorithm recurses on the low and high child of a BDD $S$.
This splits the state space into $\sspace_{|0}$ (all states starting with a 0) and $\sspace_{|1}$.
The relation $R$ can be split up accordingly as shown in \autoref{fig:partition}.
The self loops in this figure represent $S_{|0}.R_{|00}^*$ and $S_{|1}.R_{|11}^*$,
and correspond to recursive calls to \reachbdd (\autoref{line:reach00},~\ref{line:reach11}). The results of these calls need to be propagated using image computation $S_{|i}.R_{|ij} = \fname{Image}(\e S{i}, \e R{ij})$ (\autoref{line:img01}, \ref{line:img10}) until $S$ has converged, so we incorporate a loop (\autoref{line:loop}).
For notational convenience we assume $\e S{0}$ and $\e S{1}$ are program variables to which we can assign new BDDs.

\begin{figure}[t]
\centering
\subfloat[]{\begin{tikzpicture}[node distance={25mm}, auto, thick, main/.style = {draw, circle}] 
	\node[main] (0) [] {$\sspace_{|0}$};
	\node[main] (1) [right of=0] {$\sspace_{|1}$};
	\draw[->, loop left] (0) to [] node {$R_{|00}$} (0);
	\draw[->, bend right] (0) to  node[yshift=-.18em] {$R_{|01}$} (1);
	\draw[->, bend right] (1) to [] node[yshift=.18em] {$R_{|10}$} (0);
	\draw[->, loop right] (1) to [] node {$R_{|11}$} (1);
	\node[label=right:{$R = \begin{pmatrix} R_{|00} & R_{|01} \\ R_{|10} & R_{|11} \end{pmatrix}$}] (2) [below left=1cm and 0cm of 0]{};
\end{tikzpicture} \label{fig:partition-sspace}}
\subfloat[]{\begin{tikzpicture}[auto, thick,node distance=1.cm,inner sep=1.5pt]
    \tikzstyle{temp}=[main, inner sep=0, minimum size=.55cm] % to deal with prime
	\node[] (0) [label=right:{$R$}] {};
	\node[temp] (1) [node distance=.6cm,below of=0] {$x_1$};
	\node[temp] (2) [below left=.5cm and 0.3cm of 1] {$x_1'$};
	\node[temp] (3) [below right=.5cm and 0.3cm of 1] {$x_1'$};
	\node[] (00) [below left=.5cm and 0cm of 2] {$R[00]$};
	\node[] (01) [below right=.5cm and -0.4cm of 2] {$R[01]$};
	\node[] (10) [below left=.5cm and -0.4cm of 3] {$R[10]$};
	\node[] (11) [below right=.5cm and 0cm of 3] {$R[11]$};
	
	\draw[->] (0) to (1);
	\draw[->,dashed] (1) to (2);
	\draw[->] (1) to (3);
	\draw[->,dashed] (2) to (00);
	\draw[->] (2) to (01);
	\draw[->,dashed] (3) to (10);
	\draw[->] (3) to (11);
\end{tikzpicture} \label{fig:partition-dd}}
\vspace{-0.3cm}
\caption{The state space $\sspace$ can be split up into states where the first variable equals 0, and states where the first variable equals 1 (\ref{fig:partition-sspace}), and the transition relation $R$ is split up accordingly. Because source and target variables are interleaved as usual, these partitions of $R$ can be easily accessed in the BDD structure (\ref{fig:partition-dd}).}
\label{fig:partition}
\vspace{-0.5cm}
\end{figure}

%apply all transitions in each $R_{|ij}$, the recursive calls to \reachbdd and $\fname{Image}$ are repeated until the set of reachable states has converged. 
%This is necessary because the \fname{Image} calls can produce new ``seed'' states for the \reachbdd calls.

The base cases for the algorithm are as follows:
If the set of initial states or the transition relation is empty ($S=0$ or $R = 0$) there are no successors and the set of reachable states is the set of initial states. If the set of initial states contains all states ($S=1$), or if $R$ contains transitions from all states to all other states ($R=1$) and $S$ is not empty, then all states are reachable.

%These caching operations implement the top-down dynamic programming which prevents exponentially many recursive function calls for small BDDs (see \autoref{sec:prelims}).

With the decision diagram framework Sylvan~\cite{vandijk2017sylvan}, we can parallelize decision diagram operations through \fname{Spawn}/\fname{Sync} commands, which respectively fork and join light-weight tasks~\cite{vandijk2017sylvan}. However, the order of (parallel) operations is something to take into account. In particular, \autoref{line:loop1} and \ref{line:loop2} from \autoref{alg:reach-rec} are dependent, so cannot be executed in parallel. In order to parallelize \reachbdd, we change the order of calls in this loop and introduce \reachbddpar in \autoref{alg:reach-rec-par}.

%TODO: zou je ook als optimalisatie kunnen presenteren. Weten we het verschil met de andere aanpak?

{
%\vspace{-0.5cm}
\begin{algorithm}[H]
    % (algorithm to put into a algorithm2e algorithm environment)
\setcounter{AlgoLine}{5}
\While{$S$ did not converge}{
	$\fname{Spawn}(~\reachbddpar(\e S{0}, \e R{00})~)$  \Comment*[r]{Spawn call as task (fork)}
	$\e S{1} \gets \reachbddpar(\e S{1}, \e  R{11})$         \Comment*[r]{Call directly}
	$\e S{0} \gets \fname{Sync}$                      \Comment*[r]{Obtain task result (join)}
	\BlankLine
%	\tcp{In parallel:}
	$\fname{Spawn}(~\fname{Image}(\e S{1}, \e R{10})~)$  \Comment*[r]{Spawn call as task (fork)}
	$T_1 \gets \fname{Image}(\e S{0}, \e  R{01})$        \Comment*[r]{Call directly}
	$T_0 \gets \fname{Sync}$                      \Comment*[r]{Obtain task result (join)}
%
%	$S_0' \gets \fname{Image}(S_1, R_{10})$\;
%	$S_1' \gets \fname{Image}(S_0, R_{01})$\;
    \BlankLine
%	\tcp{In parallel:}
	$\fname{Spawn}(~\fname{Union}(\e S{0}, T_0)~)$  \Comment*[r]{Spawn call as task (fork)}
	$\e S{1} \gets \fname{Union}(\e S{1}, T_1)$         \Comment*[r]{Call directly}
	$\e S{0} \gets \fname{Sync}$                     \Comment*[r]{Obtain task result (join)}
%	$S_0 \gets \fname{Union}(S_0, S_0')$ \;
%	$S_1 \gets \fname{Union}(S_1, S_1')$ \;
}

\caption{\reachbddpar parallelizes the loop on \autoref{line:loop}-\ref{line:loop-end} in \autoref{alg:reach-rec}. The \fname{Image} and \fname{Union} calls are also parallelized~\cite{vandijk2015sylvan}.}
\label{alg:reach-rec-par}
\end{algorithm}
}

%This is unlike saturation, the parallelization of which is not trivial \cite{vandijk2019saturation}.

%As a side note, it is not strictly necessary to store all $S_i$ and  $R_{ij}$, as these BDDs can be obtained from $S$ and $R$ in constant time.
%\begin{algorithm}[H]
%\setcounter{AlgoLine}{5}
%\While{$S$ not converged} {
%    $S_0, S_1 \gets \fname{ChildNodes}(S)$ \;
%    \BlankLine
%    \tcc{lines \ref{line:loop-s}-\ref{line:loop-t} from Alg.\ref{alg:reach-rec}}
%    \BlankLine
%    \setcounter{AlgoLine}{10}
%    $S \gets \makenode(\bddvar{S}, S_0, S_1)$ \;
%}
%\Return{S}
%\end{algorithm}

\subsection{Analysis}
To provide some intuition for the complexity behavior of \reachbdd, we provide two cases: one where \reachbdd is exponentially faster than BFS, and one where \reachbdd reduces to BFS.
%\todo[inline]{S: The reader might wonder why we're not doing a worst case complexity analysis. Should we briefly repeat here that even just image with BDDs is \NP-complete?}

\paragraph{Ideal case}
%Based on \autoref{th:fpt-reach}, 
We give a concrete instance of a relation $R$ and an initial set of states $S_{\text{init}}$ for which \reachbdd performs exponentially better than a simple BFS.
%For this comparison, we consider the number of (top-level) \fname{Image} function calls that each algorithm needs.
Consider a transition relation $R$ which simply increases a (program) counter of $n$ bits. This counts form a starting state $S_{\text{init}} = 0 = \langle 00 \dots 0 \rangle $ in steps of 1 to $2^n-1 = \langle 11 \dots 1 \rangle$. 
As the state space is a line graph, the BFS algorithm discovers one new state every iteration, requiring $O(2^n)$ calls to the \fname{Image} function.

To illustrate the behavior of the \reachbdd algorithm, let us explicitly write all $R_{|ij}$ for $n=3$:
\begin{align*}
    R_{|00} &= 
    \begin{Bmatrix}
    \sttuple{\not\hspace{-.5mm} 000}{\not\hspace{-.5mm}001}, \\ 
    \sttuple{\not\hspace{-.5mm}001}{\not\hspace{-.5mm}010}, \\
    \sttuple{\not\hspace{-.5mm}010}{\not\hspace{-.5mm}011}
    \end{Bmatrix}
    &
    R_{|11} &=
    \begin{Bmatrix} 
    \sttuple{\not\hspace{-.5mm}100}{\not\hspace{-.5mm}101}, \\ 
    \sttuple{\not\hspace{-.5mm}101}{\not\hspace{-.5mm}110}, \\
    \sttuple{\not\hspace{-.5mm}110}{\not\hspace{-.5mm}111}
    \end{Bmatrix}
    &
    R_{|01} &=
    \begin{Bmatrix}
    \sttuple{\not\hspace{-.5mm}011}{\not\hspace{-.5mm}100}
    \end{Bmatrix}
    &
    R_{|10} &=
    \emptyset
    &
\end{align*}
While $R_{|00}$ and $R_{|11}$ represent different sets, the BDDs $\e R{00}$ and $\e R{11}$ are equal. %$R_{|00}$ is obtained from $R$ by following a low edge twice from the root node of $R$, and $R_{|11}$ by following a high edge twice. The resulting BDD of $R_{00} = R_{11}$ encodes a transition relation over one less pair of source and target variables, which are exactly the only values on which the $R_{|00}$ and $R_{|11}$ differ.

For all non-terminal cases \autoref{alg:reach-rec} does the following:

{
\setlength{\interspacetitleruled}{-.4pt}%
\begin{algorithm}[H]
\setcounter{AlgoLine}{7}
	$\e S{0} \gets \reachbdd(\e S{0}, \e R{00})$ \Comment*[r]{computes states $S_0^{\text{all}}$}
	$\e S{1} \gets \fname{Union}(\e S{1}, \fname{Image}(\e S{0}, \e R{01}))$ \Comment*[r]{generates 'seed' state $S_1^{\text{init}}$}
	$\e S{1} \gets \reachbdd(\e S{1}, \e R{11})$ \Comment*[r]{computes states $S_1^{\text{all}}$}
	$\e S{0} \gets \fname{Union}(\e S{0}, \fname{Image}(\e S{1}, \e R{10}))$ \Comment*[r]{produces no new states}
\end{algorithm}
}
\noindent
First, \reachbdd computes all reachable states which start with a 0, let us call these $S_0^{\text{all}}$. Next, the \fname{Image} call produces exactly one new state, $\langle 100 \dots 0 \rangle$, which will act as a ``seed'' state for the next \reachbdd call. 
Since the BDDs of $S_0^{\text{init}}$ and $S_1^{\text{init}}$ are equal, just as the BDDs $\e R{00}$ and $\e R{11}$, the second \reachbdd call can be looked up from cache.
Finally, since $\e R{10} = \emptyset$, no new states will be added to $S_0$, and both $S_0$ and $S_1$ will have converged.

We find two things: first, all the reachable states are found in a single loop iteration, and second, each call to \reachbdd only generates one recursive call to \reachbdd (because the second call can always be looked up from cache). Overall, \reachbdd only makes $O(n)$ recursive calls to itself and to the \fname{Image} function.

Due the monolithic nature of the transition relation, saturation will behave like BFS in this case.

\paragraph{Bad case}
%Since reachability with BDDs is proven to be outside of PSPACE \cite{bollig2010exponential}, we are mostly interested in the practical performance of reachability algorithms rather than worst-case limits. Nonetheless, in this section, we analyze the worst-case performance of \reachbdd to gain a better understanding of the procedure.
Here we provide an instance for which \reachbdd reduces to breadth-first search. %, and as such inherits certain worst-case instances from BFS.
From an arbitrary transition relation $R$ we create a new relation $R'$ for which \reachbdd behaves like to BFS. If $R$ is a relation over $2n$ variables $\{x_1, x_1', \dots x_n, x_n'\}$, we let $R'$ be a relation over $2(n+1)$ variables $\{x_0, x_0', x_1, x_1', \dots x_n, x_n'\}$. Specifically, let 
$
R' :=  x_0 \oplus {x}_0' \land R. %\{(\overline{x_0} \land x_0' \land r) \lor (x_0 \land \overline{x_0'} \land r) \mid r \in R_{\text{line}}\}
$
%The interpretation is as follows: for every original transition $r \in R$ we now have an additional variable $x_0$ which must either change from 0 to 1 or from 1 to 0. 
The corresponding decomposition into sub-functions (as visualized in \autoref{fig:partition}) looks like
\begin{align*}
    R' = \begin{pmatrix}
    0 & R'_{|01} \\
    R'_{|10} & 0
    \end{pmatrix}.
\end{align*}
%where $R'_{|00} = R'_{|11} = 0$. 
For a relation like this, the loop on \autoref{line:loop} relies entirely on the image computation steps to expand the set of reachable states, while the recursive calls never add any states.
%(although because of caching each call only takes $O(1)$ time to resolve so they are not really a hindrance either)
This effectively turns \reachbdd into BFS.
%The added variable does not affect the complexity since $O(2^{n+1}) = O(2^n)$. % O(2 \cdot 2^{n})

\subsection{MDD generalization}
\label{sec:alg-mdd-ldd}
In this section, we generalize \reachbdd (\autoref{alg:reach-rec}) from a BDD operation to an MDD operation (\autoref{alg:reach-rec-mdds}, correctness discussed in \autoref{sec:correctness}\arxiv{ and \aref{app:reach-correctness-proof}}{}).
%\alfons{leg makenode en l9/l11 uit.} 
%\sebas{\makenode~ is explained for both BDDs and MDDs in the preliminaries, do we need more explanation here? About l9/l11: I have added ``For notational convenience we assume $\e S{0}$ and $\e S{1}$ are program variables to which we can assign new BDDs.'' in 3.1, should we repeat that here?} 
For simplicity, let us assume we have a MDD encoding a set of states of $n$ variables, each of which takes values from the same domain $\mathcal{D} = \{0,1,\dots,m-1\}$, and an MDD  which encodes the transition relation which has $2n$ variables ($n$ source variables and $n$ target variables). As shown below, the state MDD can be divided into $m$ parts, and the relation MDD into $m^2$ parts, similar to how the BDDs are split up in \autoref{fig:partition}. For ease of notation we denote $m - 1 = m'$:
\begin{align*}
S &= 
\begin{pmatrix}
S_{|0} \\ S_{|1} \\ \vdots \\ S_{|m'}
\end{pmatrix}
&
R &= 
\begin{pmatrix*}[c]
R_{|00} & R_{|01} & \cdots & R_{|0,m'} \\
R_{|10} & R_{|11} & \cdots & R_{|1,m'} \\
\vdots & \vdots & \ddots & \vdots \\
R_{|m',0} & R_{|m',1} & \cdots & R_{|m',m'}
\end{pmatrix*}
\end{align*}

Where the BDD algorithm iterates over four transition relations $R_{|ij}$, the MDD algorithm simply iterates over all $m^2$ relations $R_{|ij}$. When $i = j$, $R_{|ij}$ contains transitions for which the first variable stays the same, and we can call $\reachmdd(S_{|i}, R_{|ij})$. For all cases where $i \neq j$ we use $\fname{Image}(S_{|i}, R_{|ij})$ instead. 
%The result is shown in \autoref{alg:reach-rec-mdds}.

%To avoid keeping track of $m$ separate state MDDs during the while loop, we only keep track of $S$ and update it during the loop. This can be done by taking the newly found states $T$, extending them with the ``missing'' variable $\bddvar{S}$, and taking the union of the result with $S$.

%{
%\setlength{\interspacetitleruled}{0pt}
%\setlength{\algotitleheightrule}{0pt}
%\begin{algorithm}[H]
%\SetKwFunction{UpdateMDD}{Update}
%\Fn{\UpdateMDD{$S, s_j, j$}}{
%    $s_j' \gets x_j \land s_j $\; %\fname{MakeNodeMDD}(j, s_j, 0)$\;
%	\Return{$\fname{Union}(S,s_j')$}
%}
%\end{algorithm}
%}
%
%\begin{figure}
%    \centering
%
%    \caption{The $\fname{Update}(S,s,j,k)$ creates an MDD $s'$ with top variable $x_k$, from which the $x_k = j$ edge points to $s$, and all other edges to 0. This brings $S$ and $s'$ into the same variable domain, so they can be conjoined using \fname{Union}.}
%    \label{fig:mdd-update}
%\end{figure}
%

\begin{algorithm}[H]
\SetKwFunction{ReachMDD}{ReachMdd}
% (algorithm to put into a algorithm2e algorithm environment)
\Fn{\ReachMDD{$S$, $R$}}{
     \vspace{-1.2em}\Comment*[r]{For MDDs $S,R$ on $n,2n$ variables.}
	\BlankLine
	\lIf{$S = 0$}{\Return{0}}
	\lIf{$R = 0$}{\Return{$S$}}
	\lIf{$S = 1$}{\Return{1}}
	\lIf{$R = 1 \land S \neq 0$}{\Return{1}}
	\BlankLine
	\label{line:missing-cache-lookup-mdd}
    
	\While{$S$ did not converge}{
		\For{$i,j \in \mathcal{D}$}{
			\If{$i = j$} {
				$\e S{i} \gets \reachmdd(\e S{i}, \e R{ij})$ \;
			}
			\Else{
			    $\e S{j} \gets \fname{Union}(\e S{j}, \fname{Image}(\e S{i}, \e R{ij}))$ \;
			}
			
%			$T' \gets \fname{MakeNode}(\bddvar{S}, \set{j \rightarrow T} )$
%			 ~~\Comment{Create node:
%			 \raisebox{-.8cm}{\scalebox{.8}{ \hspace{-.55cm}
%			 \begin{tikzpicture}[node distance=3mm and 1mm, thick, edgelabel/.style={fill=white,rounded corners=4pt,inner sep=2pt}]
%	    \node[] (root) [] {};
%	    \node[main] (0) [node distance=2mm,below=of root]  {$x$};
%	    \node[] (1) [below left=of 0] {$T$};
%	    \node[leaf] (terminal) [below right=of 0] {0};
%	    \draw[->] (root) to node [auto] {} (0);
%	    \draw[->, bend right] (0) to node[above left,edgelabel] {$j$} (1);
%	    \draw[->, bend left] (0) to node[above right,edgelabel] {$\neq j$} (terminal);
%           \end{tikzpicture}
%                }\hspace{-.6cm}}\vspace{-.8cm}}
%            $S \gets \fname{Union}(S, T')$
		}
	}
	\label{line:missing-cache-put-mdd}
	\BlankLine
	\Return{$\makenode(\bddvar{S}, \e S{0}, \dots, \e S{m-1})$} \label{line:mdd-makenode}
}

\caption{An MDD operation for computing reachability. Cache lookup/insert for dynamic programming after \autoref{line:missing-cache-lookup-mdd} and \ref{line:missing-cache-put-mdd} are omitted.}
\label{alg:reach-rec-mdds}
\end{algorithm}

We note that \reachbdd does not generalize so well to MDDs, in the following sense: for BDDs, \autoref{alg:reach-rec} has two \reachbdd calls and two \fname{Image} calls inside the loop. However for MDDs, we get $O(m)$ \reachmdd calls and $O(m^2)$ \fname{Image} calls every loop iteration. In the MDD case, a larger part of the computation is no longer handled by recursive calls, but instead by image computations.

\subsection{Correctness}
\label{sec:correctness}
In this section we give a sketch of the correctness proofs for both \reachmdd (\autoref{th:reachmdd-proof}) and \reachbdd (\autoref{th:reachbdd-proof}). A complete proof can be found in  \arxiv{\aref{app:reach-correctness-proof}}{\cite[App. A]{arxivversion}}.

\begin{restatable}[]{theorem}{reachmddproof}
\label{th:reachmdd-proof}
Given two MDDs $S$ and $R$ with $n$ and $2n$ variables respectively, $\reachmdd(S,R)$ (\autoref{alg:reach-rec-mdds}) computes all the reachable states $S.R^*$.
\end{restatable}

\begin{proofsketch}
The correctness of \reachmdd can be shown by means of algorithm transformation from breadth-first search. The algorithm for BFS (given below) directly follows from the definition $S.R^* = \Union_{k=0}^\infty S.R^k$, as shown by the Knaster-Tarski theorem \cite{tarski1955lattice}.

{
\setlength{\interspacetitleruled}{0pt}
\setlength{\algotitleheightrule}{0pt}
\begin{algorithm}[H]
    \While{$S$ did not converge} {
        $S \gets \fname{Union}(S, \fname{Image}(S, R))$
    }
    \Return{$S$}
\end{algorithm}
}
\noindent
The two main steps in the transformation are as follows: first, using the Shannon decomposition, the computation of $\fname{Image}(S, R))$ can be split up into calls $\fname{Image}(S[i], R[ij]))$ for all $i,j\in \mathcal{D}$, the results of which can be combined with a $\fname{MakeNode}$ function as on \autoref{line:mdd-makenode} of \autoref{alg:reach-rec-mdds}.
Second, since we are ultimately computing $S.R^*$, the calls $\fname{Image}(S[i], R[ij]))$ can be replaced with $\reachmdd(S[i],R[ij])$ when $i = j$.
The algorithm \reachmdd follows directly from these two steps.
\end{proofsketch}

\begin{restatable}[]{corollary}{reachbddproof}
\label{th:reachbdd-proof}
Given two BDDs $S$ and $R$, $\reachbdd(S,R)$ (\autoref{alg:reach-rec}) computes all the reachable states $S.R^*$.
\end{restatable}

\begin{proofsketch}
The correctness of \reachbdd can be shown from \autoref{th:reachmdd-proof} by taking an MDD with $\mathcal{D} = \{0,1\}$.
\end{proofsketch}

\section{Empirical evaluation}
\label{sec:experiments}

\subsection{Experimental Setup}
We implemented the new algorithms in the decision diagram package Sylvan \cite{vandijk2013multi,vandijk2015sylvan,vandijk2017sylvan}, using the task-based scheduling as described in \autoref{sec:alg-bdd}. Instead of MDDs, Sylvan supports LDDs (see \autoref{sec:mdds}), which can be seen as a particular implementation of MDDs.
For our experiments, we compare against the saturation procedure for BDDs and LDDs as implemented in Sylvan~\cite{vandijk2019saturation}.\footnote{\anonymize{The implementation of our algorithms, along with the repeatable experiments can be found here: \url{https://github.com/sebastiaanbrand/reachability}}{The code will be made available for the FM 2023 artefact submission, and will also be made publicly available. This to ensure anonymity.}}

% Dataset
\begin{wraptable}[6]{r}{6cm}\vspace{-.7em}
    \centering
    \caption{Overview of used benchmarks}\vspace{.4em}
    \label{tab:benchmarks}
    %\vspace{-0.5em}
\scalebox{.85}{
        \begin{tabular}{|c|c|c|}
        \hline
        \textbf{Source} & \textbf{Type} & \textbf{\#} \\\hline
        BEEM \cite{pelanek2007beem} & DVE & 300  \\
        MCC`16 \cite{mcc2016} & Petri-nets & 357 \\
        SPINS \cite{van2013spins} & Promela & 35 \\\hline
    \end{tabular}
}
\end{wraptable}

We use a number of existing benchmark sets to evaluate the performance of our algorithms. Specifically, we use the BEEM benchmark set, consisting of 300 instances of models in the DVE language, a benchmark set of over 300 Petri nets from the Model Checking Contest 2016 (MCC~2016) \cite{mcc2016}, and a small set of Promela models, compiled for the SpinS extension of LTSmin \cite{van2013spins}. 
%Most of these Promela models were originally put together by Alberto Lluch Lafuente, although the web page in question is not online anymore. 
%For all the complete benchmark sets we refer to the repository at \cite{brand2021repo}.

For these benchmarks we use the same experimental setup as used in \cite{vandijk2019saturation}: we first use the model checker LTSmin \cite{kant2015ltsmin} to generate BDDs and LDDs for the (partial) transition relations and initial set of states, which are exported to \texttt{.bdd} and \texttt{.ldd} files. In LTSmin, the partial transition relations are ``learned'' on-the-fly, during the exploration~\cite{kant2015ltsmin}, as opposed to directly building the transition relations from a modeling language like NuXMV~\cite{cavada2014nuxmv}.
The variables in these BDDs and LDDs have been reordered by LTSmin with Boost's Sloan algorithm, since this reordering strategy has shown good results for saturation 
\cite{meijer2016bandwidth,amparore2017gradient,amparore2018decision}.

Since we are interested in comparing the \reach algorithms against saturation, we require a single transition relation for \reach. Therefore, as a preprocessing step, for the \reach algorithms only, we merge the partial relations from LTSmin into a monolithic transition relation over all variables. The run time of the merging of partial relations is included in the total run times reported in \autoref{fig:sat-vs-rec}.
This approach is slightly disadvantaging \reach because we could also change the setup to generate monolithic relations directly (as we do for Petri nets in the comparison with ITS-tools in \autoref{sec:experiments-results}), but using the setup from \cite{vandijk2019saturation} allows us to make a direct comparison with parallel saturation from \cite{vandijk2019saturation}.

%The difference between the on-the-fly and offline symbolic exploration of the state space has been studied in~\cite{vandijk2019saturation}. Here we are interested in comparing the \reach operations against saturation, and hence require a single  transition relation for \reach. Therefore, as a preprocessing step, we merge the partial relations from LTSmin  into a monolithic transition relation over all variables. 
%This setting is a little artificial, as normally the transition relations are built directly from a modeling language like NuXMV~\cite{cavada2014nuxmv}, but it allows us to compare \reach against the same large model set that was used in~\cite{vandijk2019saturation}.
%To compensate for any discrepancies that might occur due to this experimental setup, we include the run time of the merging of partial relations in the run times reported in \autoref{fig:sat-vs-rec} (disadvantaging the \reach operations).
%The time used to merge the partial relations is not included in the reported reachability time.

We limit the run time of each reachability method to 10 minutes. 
%The implementations can be found at \cite{brand2021repo}.
The sequential benchmarks were performed on a machine with an AMD Ryzen 7 5800x CPU and 64 GB of available memory. The parallel benchmarks on a machine with 4 Intel Xeon E7-8890 v4 CPUs with 24 physical cores each (96 in total) and 2 TB of memory. Aside from the experiments which test parallelism specifically, all reported run times are for a single core.

\subsection{Results}
\label{sec:experiments-results}
\paragraph{Comparison with saturation}
A comparison between the runtimes of saturation and \reach is given in \autoref{fig:sat-vs-rec}.
In this discussion we differentiate between smaller models (run times $\leq$1 sec) and larger models (run times $\geq$1 sec).
For BDDs, \reach outperforms saturation on a large number of bigger DVE models, but encounters timeouts as well.
For LDDs, \reach appears competitive with saturation on DVE and Promela models, while the trend for the larger Petri net models shows \reach outperforming saturation by up to a factor 100.
For both BDDs and LDDs saturation is often faster on the smaller instances. This is in part due to the fact that the relative overhead of merging the partial relations is greater for smaller instances.
%Looking at the specific Petri net instances on which saturation performs better than our method we can see a pattern.
%We validate the correctness of the results from \reach against the results from the existing saturation implementation.

%Out of all benchmarks which finished under 10 minutes ($\sim$450), our \reachmdd implementation reported incorrect state counts on 16 instances. These have been show as timeouts in \autoref{fig:sat-vs-rec}. To determine correct state counts we use the state counts from the saturation implementation \cite{vandijk2019saturation} as reference. All state counts produced by \reachbdd agree with saturation. While unable to find the issue with \reachmdd, we believe this is an error in the implementation, and not the algorithm itself.

\begin{figure}[t]
%\hspace{-.5cm}
\centering
\includegraphics[width=0.95\textwidth]{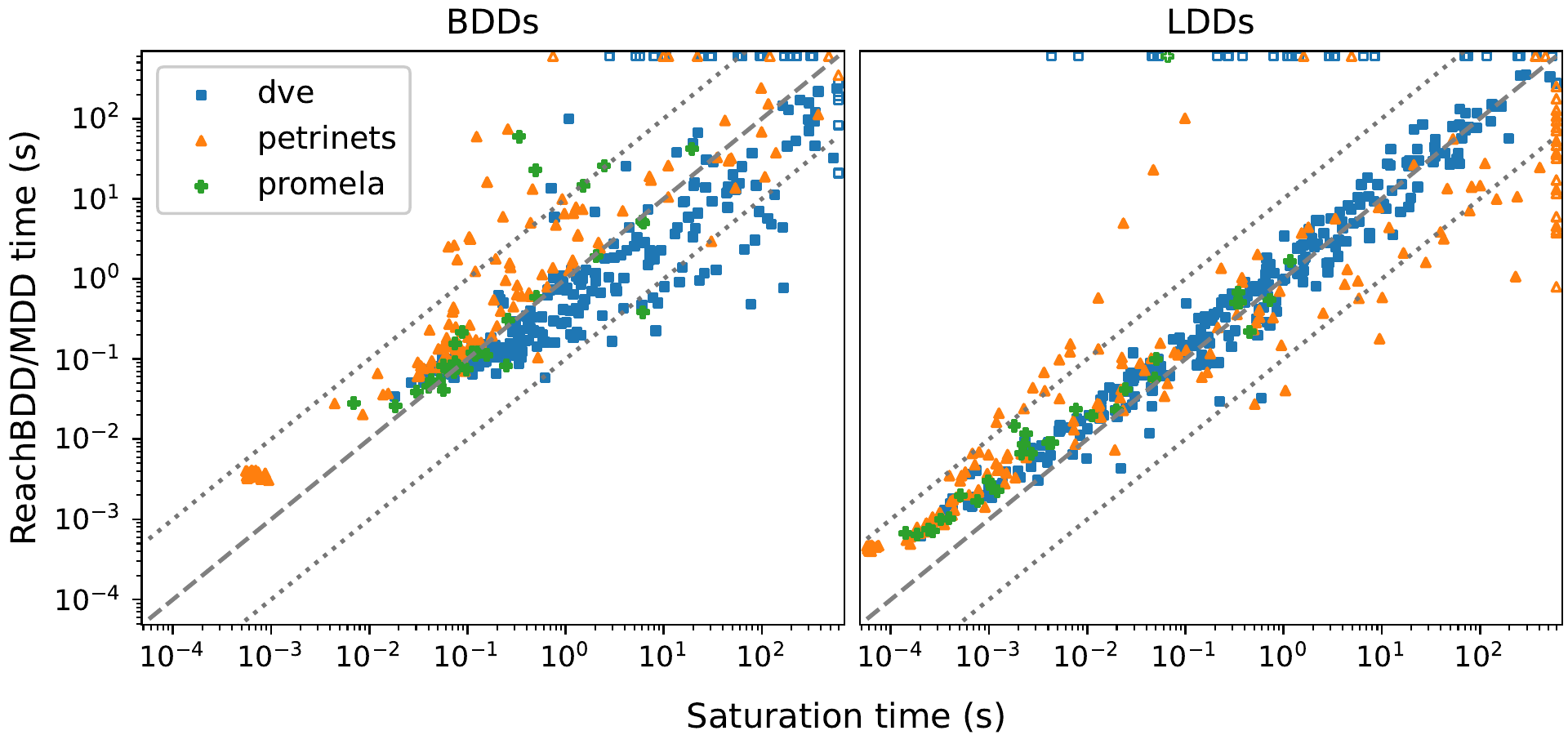}
\vspace{-.2cm}
\caption{The run time of finding all reachable states with BDDs (left) and LDDs (right) using  \reachbdd and \reachmdd versus saturation. 
%The dashed line indicates equal performance, while the dotted lines indicate a difference of factor $\times 10$ and $\times \tfrac{1}{10}$. 
Open markers indicate timeouts.
}
\label{fig:sat-vs-rec}
\end{figure}

\paragraph{Locality}
As discussed in \autoref{sec:prelim-saturation}, saturation is known to work well on transition systems where the partial transition relations exhibit locality. To get insight into how locality affects our new algorithms relative to saturation, we define the \textit{average relative bandwidth} as a metric for locality:
For $k$ partial relations $R_1, \dots, R_k$, sorted in an ascending order based on their first variable,  we can define a $k \times k$ matrix $M$ with entries $m_{ij}$ such that $m_{ij} = 1$ if $R_i$ and $R_j$ share at least one variable, and 0 otherwise. We define the bandwidth of a row $m_{i,*}$ as the distance between the first and the last non-zero element in this row. 
The average bandwidth is then simply the average of the bandwidths of all the rows $m_{i,*}$. The average \emph{relative} bandwidth is the average bandwidth divided by $k$. Note that this $k \times k$ matrix is different from (although related to) the $k \times n$ variable matrix shown in \autoref{sec:prelim-saturation}. The matrix $M$ and the locality metric derived from it are independent of the variable order in the decision diagram.

Plotting the run time of \reachbdd divided by the run time of saturation against this average relative bandwidth (\autoref{fig:relative-bandwidth}), we see that there is a negative correlation. Although not extremely strong, this correlation shows that the benchmarks on which saturation outperforms our algorithms are predominantly the instances where the partial relations are relatively local, while on instances with less locality our algorithms have a greater edge over saturation.

\begin{wraptable}[7]{r}{6.0cm}\vspace{-1.0em}
    \centering
    \vspace{-2.2em}
    \caption{Parallel speedups}
    \setlength\tabcolsep{0.35em}
\scalebox{.85}{
        \begin{tabular}{|c|c|ccc|}
        \hline
        \multirow{2}{*}{\textbf{algorithm}} & \multirow{2}{*}{\textbf{cores}} & \multicolumn{3}{c|}{\textbf{speedup}} \\
         & & $\mathbf{P_{95}}$ & $\mathbf{P_{99}}$ & $\mathbf{P_{99.5}}$\\\hline
        saturation \cite{vandijk2019saturation}  & 16 & $\times 8.1$ & $\times 11$ & $\times 11$\\
        \reachbddpar & 16 & $\times 6.6$ & $\times 8.3$ & $\times 8.8$ \\\hline
        saturation \cite{vandijk2019saturation}  & 64 & $\times 8.7$ & $\times 22$ & $\times 22$ \\
        \reachbddpar & 64 & $\times 5.6$ & $\times 9.1$ & $\times17$ \\\hline
    \end{tabular}
}
\end{wraptable}
\paragraph{Parallelism}
\autoref{fig:parallel-bench} shows the speedups obtained by \reachbddpar and the parallelized version of saturation from \cite{vandijk2019saturation} on 16 and 64 cores. The table on the right gives the $95^{\text{th}}$, $99^{\text{th}}$ and $99.5^{\text{th}}$ percentile of the speedups.
We see that for the 16 core runs \reachbddpar is able to keep up with \cite{vandijk2019saturation}, although falling slightly behind. For the 64 core runs, while \reachbddpar falls behind \cite{vandijk2019saturation} on the $99^{\text{th}}$ percentile, it is still able to achieve a $\times 17$ speedup on in the $99.5^{\text{th}}$ percentile, compared to \cite{vandijk2019saturation}'s $\times 22$.

\paragraph{Comparison against ITS-tools}
We also briefly compare how \reach performs against a state-of-the-art model checking tool. For this we pick ITS-tools \cite{thierry2015symbolic}, the overall highest scoring tool in the Model Checking Contest 2021 \cite{mcc2021}. 
Since here we compare against a different tool, as opposed to comparing algorithms within the same package, we need to slightly extend our setup.
We add two things: first we create a small program \texttt{pnml-encode} which builds the decision diagrams of the transition relations directly from the Petri net files. Second, we extend our LDDs with a (much simpler) version of homomorphisms which are also used in the set decision diagrams (SDDs) \cite{couvreur2005hierarchical}, which are a part of ITS-tools.

The results are given in \autoref{fig:its-comparison}. While ITS-tools outperforms \reachmdd on average, there is a significant number of instances where ITS-tools gives timeouts and \reach does not. Including these timeouts, \reach is faster than ITS-tools on 29\% of instances. This suggest that \reach could be useful as a complementary method in an ensemble tool, where a different method can be tried if the first one times out.

\begin{figure}[H]
\vspace{-0.3cm}
\centering
\includegraphics[width=0.85\textwidth]{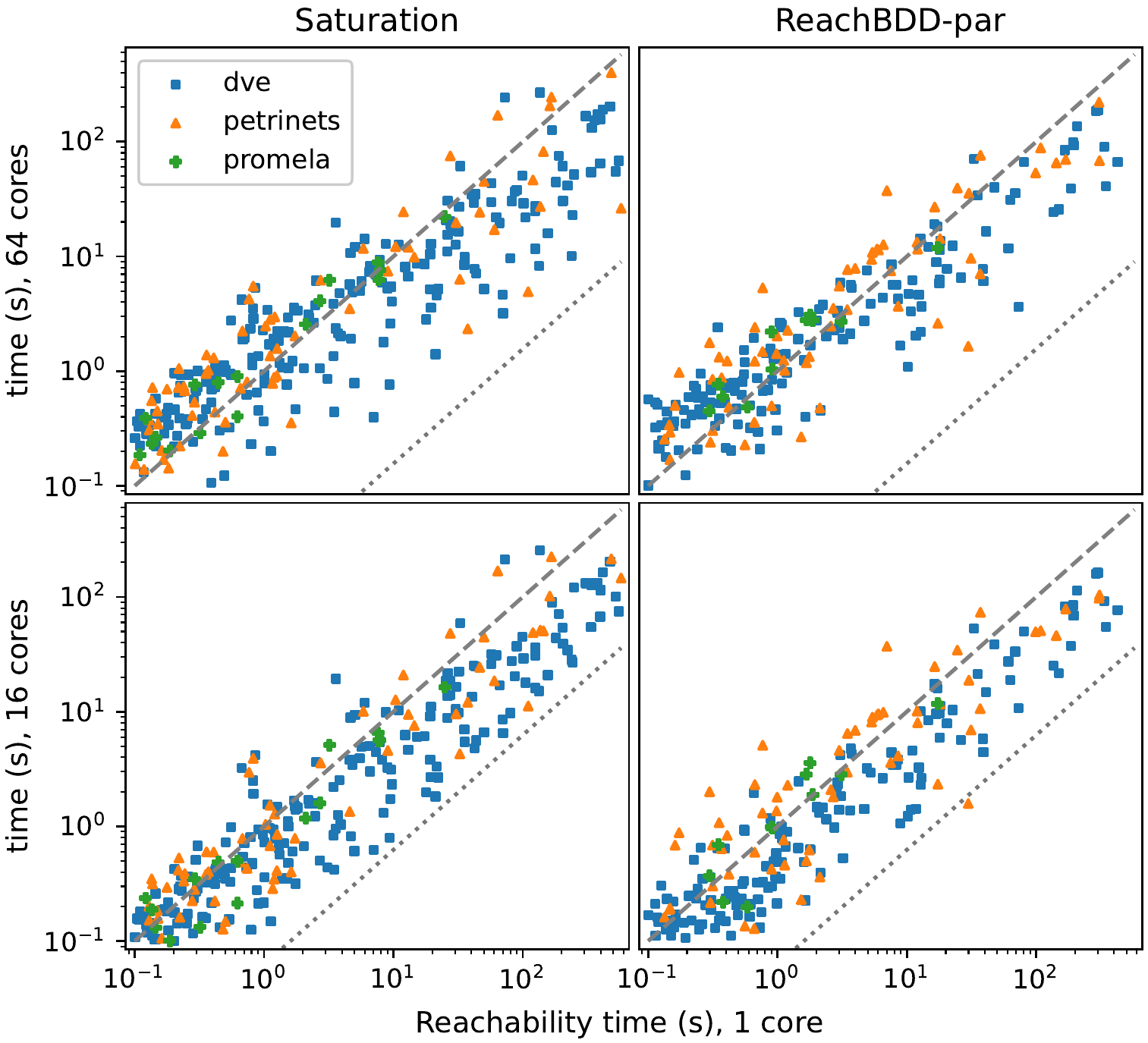}
\vspace{-0.4cm}
\caption{Parallel speedup for saturation (left column) and \reachbddpar (right column). The dotted diagonal lines indicates a speedup of a factor 16 (bottom row) and 64 (top row) relative to the single core performance.}
\label{fig:parallel-bench}
\end{figure}

\begin{figure}[H]
\vspace{-1.8cm}
\begin{minipage}{\linewidth}
      \centering
      \begin{minipage}{0.47\linewidth}
          \begin{figure}[H]
              \includegraphics[width=\linewidth]{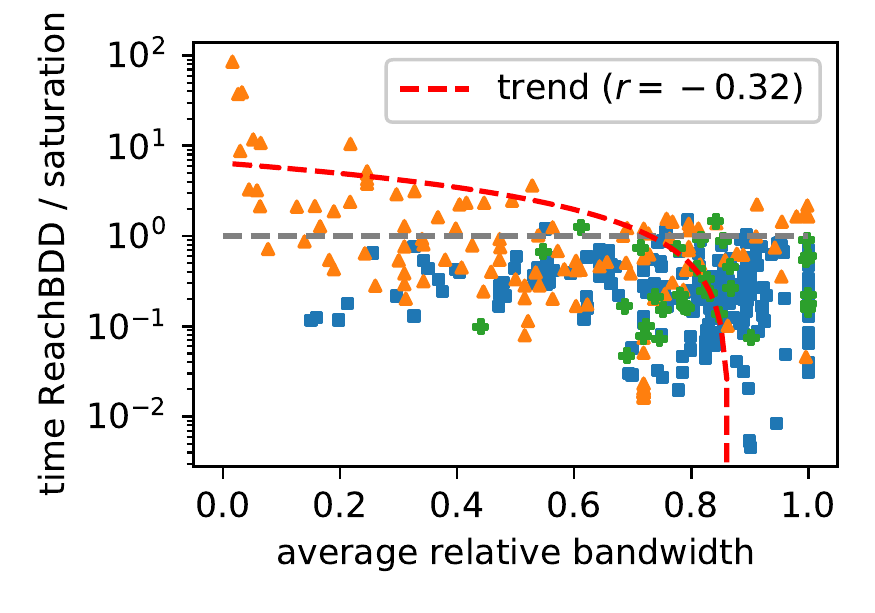}
              \vspace{-0.77cm}
              \caption{The effect of locality on the relative performance of \reachbdd. For \reachmdd $r=-0.11$.
              }
              \label{fig:relative-bandwidth}
          \end{figure}
      \end{minipage}
      \hspace{0.02\linewidth}
      \begin{minipage}{0.47\linewidth}
          \begin{figure}[H]
              \includegraphics[width=\linewidth]{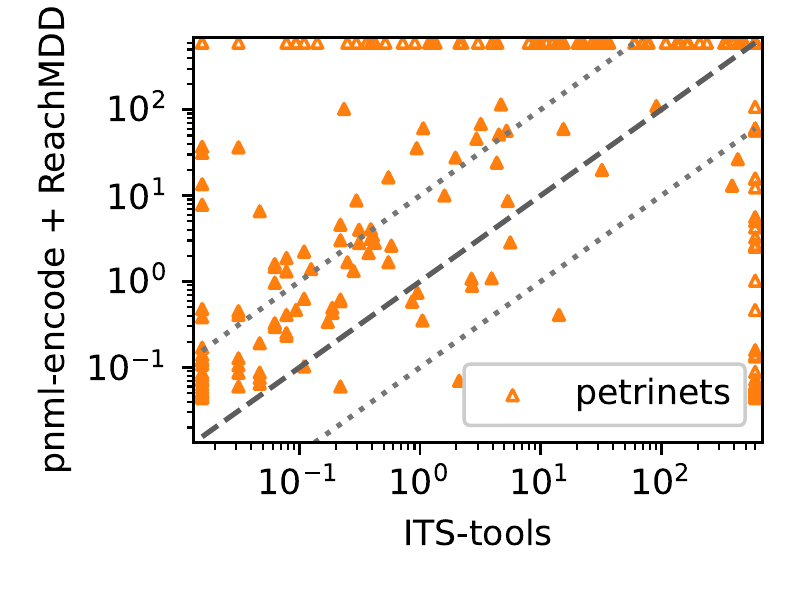}
              \vspace{-1.1cm}
              \caption{Comparison of \reachmdd against computing reachable states with ITS-tools.
              }
              \label{fig:its-comparison}
          \end{figure}
      \end{minipage}
\end{minipage}
\end{figure}

\section{Conclusion}
\label{sec:conclusion}
\paragraph{Summary}
We presented two new reachability operations for decision diagrams: \reachbdd and \reachmdd. In contrast to other approaches, like saturation, these operations can act on a single monolithic transition relation. Similar to saturation, these new algorithms build the decision diagram for the reachable states (at least partially) bottom-up. One advantage of these operations is their simplicity. This simplicity allows us for example to more easily %leverage Sylvan's capabilities of 
parallelize the decision diagram operations, as demonstrated for \reachbddpar. 
Empirical evaluation of \reachbdd and \reachmdd on a large number of benchmark sets shows that the new operations are competitive with saturation, and tend to outperform saturation on larger instances.
Additionally, we find that \reachbddpar's peak parallel performance does not fall far behind that of saturation.
Finally, our empirical results show that the \reach operations
	can solve 29\% of instances faster than ITS-tools, 
	which indicates \reach can be useful as a complementary algorithm.

% FPT claim
%We also showed that the new algorithms are fixed parameter tractable (FPT),
%despite the fact that graph reachability with BDDs as inputs
%is PSPACE-hard~\cite{feigenbaum1998complexity}.
%We also show that image computation, used in the \reachbdd operation, 
%is FPT (also in the size of the bandwidth of the transition relation matrix).
%While image computation can easily encode the satisfiability\todo{rephrase / move} problem~\cite{mcmillan1992symbolic}
%by exploiting its unbounded existential quantification (see \autoref{eq:image}),
%the bounded bounded bandwidth necessarily also has bounded \concept{treewidth}, and satisfiability is FPT in the tree width~\cite{downey2012parameterized}.

% future work
\paragraph{Future work}
Saturation still outperforms our algorithms on a number of instances. Many of these instances have a lot of locality, which is exactly the regime where saturation is expected to do very well. With further investigation, $\fname{ReachBdd}$ and $\fname{ReachMdd}$ could potentially be modified to perform better on such instances. The current bottleneck, as illustrated by our analysis, is the reliance on the standard \fname{Image} operation. Further integration of both operations could perhaps yield improvement.

%The operations $\fname{ReachBdd}$ and $\fname{ReachMdd}$ are currently (by design) defined and implemented for a single global transition relation. However, in principle these algorithms can also be adapted to work on partial transition relations (relations which act on only a subset of the state space variables), which would allow them to be integrated into an overarching reachability strategy which acts on a number of partial relations, like chaining or saturation.
% - allow this operation to work on partial relations, maybe integrating it in other reachability algorithms for partial relations (like saturation) works well

\subsubsection*{Acknowledgements.}
This work was supported by the NEASQC project, funded by European Union’s Horizon 2020, Grant Agreement No. 951821.

% Bibliography
\typeout{} % bibliography fix on overleaf
\bibliographystyle{splncs04}
\bibliography{references}

% only include appendix in arXiv version
\arxiv{
\newpage
\appendix
\section{REACH correctness proof}
\label{app:reach-correctness-proof}

% Some custom algorithm2e settings for appendix
\SetCustomAlgoRuledWidth{\textwidth}
\LinesNumberedHidden
\newcommand{\algpadding}{\vspace{0.05cm}}

This appendix contains the correctness proofs for \reachbdd and \reachmdd.
Since \reachbdd can be seen as a special case of \reachmdd where $\mathcal{D} = \set{0,1}$, we first prove the correctness of \reachmdd in \autoref{th:reachmdd-proof} and then show that this also holds for \reachbdd.

Normally the correctness of a decision diagram operation which implements a logical operation $\circ$ is derived by expanding the arguments via the Shannon decomposition and rewriting the logical formula for $\circ$ (see for example \cite[Eq.~1]{bryant1992symbolic}). However, to avoid introducing notation for the reflexive-transitive closure ($*$) in logical formulas, we instead choose to present a more intuitive proof on an algorithmic level.

\reachmddproof* % restated theorem from main text

\begin{proof}
Here we will show that \reachmdd can be derived by means of algorithm transformation from breadth-first search (BFS), which in turn follows directly from the definition of $S.R^*$.

%\todo[inline]{We could write an inductive proof (induction over the number of variables), showing that the algorithm indeed returns a DD that represents the desired result (See for an example [.., Th xx]). Here we opt instead to derive the algorithm through transformation of the standard reachability algorithm that uses DD operations (which?). We choose this approach because the \reach algorithm can be derived with basic transformations (true?), and because we think it is a little more intuitive than an inductive approach.}

Following the definition $S.R^* = \Union_{k=0}^\infty S.R^k$, as shown by the Knaster-Tarski theorem \cite{tarski1955lattice}, the reachable states can be computed with the algorithm below.% The guarantee that this BFS algorithm eventually terminates follows from the Knaster-Tarski theorem \cite{tarski1955lattice}.

{
\algpadding
\setlength{\interspacetitleruled}{0pt}
\setlength{\algotitleheightrule}{0pt}
\begin{algorithm}[H]
    \While{$S$ did not converge} {
        $S \gets \fname{Union}(S, \fname{Image}(S, R))$
    }
    \Return{$S$}
\end{algorithm}
\algpadding
}

\noindent
The $\fname{Image}$ call can be split into calls on pairs $(S_{|i}, R_{|ij})$, $\forall i, j \in \mathcal{D}$. This is the Shannon decomposition on an algorithm level.

{
\algpadding
\setlength{\interspacetitleruled}{0pt}
\setlength{\algotitleheightrule}{0pt}
\begin{algorithm}[H]
    \While{$S$ did not converge} {
        \For{$i, j \in \mathcal{D}$}{
            $\e S{j} \gets \fname{Union}(\e S{j}, \fname{Image}(\e S{i}, \e R{ij}))$
        }
    }
    \Return{$\fname{MakeNode}(\bddvar{S}, \e S{0}, \dots, \e S{|\mathcal{D}|-1})$}
\end{algorithm}
\algpadding
}

\noindent
Because the loop is executed until $S$ converges, the algorithm above is equivalent to the following algorithm.

{
\algpadding
\setlength{\interspacetitleruled}{0pt}
\setlength{\algotitleheightrule}{0pt}
\begin{algorithm}[H]
    \While{$S$ did not converge} {
        \For{$i, j \in \mathcal{D}$}{
            \If{$i = j$} {
                \While{$\e S{i}$ not converged}{
                    $\e S{i} \gets \fname{Union}(\e S{i}, \fname{Image}(\e S{i}, \e R{ij}))$
                }
            }
            \Else{
                $\e S{j} \gets \fname{Union}(\e S{j}, \fname{Image}(\e S{i}, \e R{ij}))$ \\
            }
        }
    }
    \Return{$\fname{MakeNode}(\bddvar{S}, \e S{0}, \dots, \e S{|\mathcal{D}|-1})$}
\end{algorithm}
\algpadding
}

\noindent
%The while loop over $\e S{i}$ computes $S_{|i}.R_{|ii}^*$, which are the reachable states for MDDs $\e S{i}$ and $\e R{ii}$ with $n-1$ and $2(n-1)$ variables respectively. 
%Assuming \reachmdd correctly computes the reachable states for a domain of $n-1$ variables, we can substitute this loop and obtain:
Now we can identify the inner while loop as a recursive call to \reachmdd based on its resemblance to the initial BFS procedure.
In the following algorithm, we therefore replace the inner loop with a recursive call, inductively solving the problem on MDDs $\e S{i}$ and $\e R{ij}$ with $n-1$ and $2(n-1)$ variables respectively.

{
\algpadding
\setlength{\interspacetitleruled}{0pt}
\setlength{\algotitleheightrule}{0pt}
\begin{algorithm}[H]
    \While{$S$ did not converge} {
        \For{$i, j \in \mathcal{D}$}{
            \If{$i = j$} {
                $\e S{i} \gets \reachmdd(\e S{i}, \e R{ij})$ \\
            }
            \Else{
                $\e S{j} \gets \fname{Union}(\e S{j}, \fname{Image}(\e S{i}, \e R{ij}))$ \\
            }
        }
    }
    \Return{$\fname{MakeNode}(\bddvar{S}, \e S{0}, \dots, \e S{|\mathcal{D}|-1})$}
\end{algorithm}
\algpadding
}
%This construction of \reachmdd for a domain of $n$ variables relies on the assumption that \reachmdd indeed returns the reachable states for a domain of $n-1$ variables. Together with the terminal cases which are treated below, this assumption holds by induction over $n$.

Finally, we treat the terminal cases of \reachmdd. The following are cases where were at least one of $S$ and $R$ is a leaf node.
\[
\reachmdd(S, R) = S.R^* = 
\begin{cases}
0 & \text{if } S = 0 \\
S & \text{if } R = 0 \\
1 & \text{if } S = 1 \\
1 & \text{if } R = 1 \text{ and } S \neq 0
\end{cases}
\]
The terminal cases are derived as follows. If the set of initial states is empty ($S=0$) then the set of reachable states is also empty. If the transition relation is empty ($R = 0$) then, because we consider the \textit{reflexive}-transitive closure, the set of reachable states is the set of initial states. If the set of initial states contains all states ($S=1$) then the set of reachable states also contains all states (because of reflexivity this holds even when $R=0$). Finally, if the transition relation contains all transitions from any state to any other state ($R=1$) then, if $S$ is not empty, all states are reachable.

\qed
\end{proof}

\reachbddproof* % restated corollary in main text

\begin{proof}
Filling in $\mathcal{D} = \set{0,1}$ in \reachmdd and expanding the individual loop iteration we obtain the following.

{
\algpadding
\setlength{\interspacetitleruled}{0pt}
\setlength{\algotitleheightrule}{0pt}
\begin{algorithm}[H]
    \While{$S$ did not converge} {
        $\e S{0} \gets \reachbdd(\e S{0}, \e R{00})$ \\
        $\e S{1} \gets \fname{Union}(\e S{1}, \fname{Image}(\e S{0}, \e R{01}))$ \\
        $\e S{0} \gets \fname{Union}(\e S{0}, \fname{Image}(\e S{1}, \e R{10}))$ \\
        $\e S{1} \gets \reachbdd(\e S{1}, \e R{11})$
    }
    \Return{$\fname{MakeNode}(\bddvar{S}, \e S{0}, \e S{1})$}
\end{algorithm}
\algpadding
}
\noindent
Since each step only adds elements to $S$, and the loop runs until $S$ converges, the order of the steps within the loop can be freely permuted, which gives us exactly \autoref{alg:reach-rec}.
\qed
\end{proof}

}{}

\end{document}